\documentclass[a4paper, english, hidelinks]{lipics-v2018}

\usepackage{microtype}
\usepackage{stmaryrd}
\usepackage{wrapfig}
\usepackage{bbm}
\usepackage{thm-restate}

\newcommand\inter[1]{\llbracket #1 \rrbracket}

\newcommand\delay{\text{del}}
\newcommand\aut{\mathcal{A}}
\newcommand\tra{\mathcal{T}}

\newcommand\inp{\mathbbmtt{i}}
\newcommand\outp{\mathbbmtt{o}}
\newcommand\dom{\text{dom}}

\newcommand{\alphint}[1]{I_{#1}}
\newcommand{\nojp}[1]{\psi_{#1}}
\newcommand{\prodalp}[1]{\Sigma_{#1}}
\newcommand{\produ}[1]{\psi_{#1}}
\newcommand{\swap}{f}

\newcommand\out{\text{out}}
\newcommand\del{\text{del}}

\newcommand{\pick}[1]{\phi(#1)}

\usepackage{graphicx}
\usepackage{xcolor,pgf,tikz}
\usetikzlibrary{arrows,automata,positioning,decorations.pathmorphing,snakes,shapes}

\usepackage{ifthen}
\usepackage{xargs}

\newcommandx{\yaHelper}[2][1=\empty]{%
\ifthenelse{\equal{#1}{\empty}}%
  { \ensuremath{ \scriptstyle{ #2 } } } 
  { \raisebox{ #1 }[0pt][0pt]{ \ensuremath{ \scriptstyle{ #2 } } } }  
}   

\newcommandx{\yrightarrow}[4][1=\empty, 2=\empty, 4=\empty, usedefault=@]{%
  \ifthenelse{\equal{#2}{\empty}}
  { \xrightarrow{ \protect{ \yaHelper[ #4 ]{ #3 } } } } 
  { \xrightarrow[ \protect{ \yaHelper[ #2 ]{ #1 } } ]{ \protect{ \yaHelper[ #4 ]{ #3 } } } } 
}

\newcommand{\myxrightarrow}[1]{\yrightarrow{#1}[-2pt]}

\newtheorem{proposition}[theorem]{Proposition}

\title{The Complexity of Transducer Synthesis from Multi-Sequential Specifications}
\titlerunning{Transducer Synthesis from Multi-Sequential Specifications}

\author{L\'eo Exibard}{Aix-Marseille Universit\'e, Marseille, France \\ Universit\'e libre de Bruxelles, Brussels, Belgium}{leo.exibard@ulb.ac.be}{}{L. Exibard is a PhD student funded by a FRIA fellowship from the F.R.S.-FNRS.}

\author{Emmanuel Filiot}{Universit\'e libre de Bruxelles, Brussels, Belgium}{efiliot@ulb.ac.be}{}{E. Filiot is a research associate of F.R.S.- FNRS. He is supported by the ARC Project Transform F\'ed\'eration Wallonie-Bruxelles and the FNRS CDR project J013116F.}
\author{Isma\"el Jecker}{Universit\'e libre de Bruxelles, Brussels, Belgium}{ismael.jecker@ulb.ac.be}{}{I. Jecker is an ``aspirant FNRS'' PhD student, funded by the F.R.S.-FNRS.}

\Copyright{L\'eo Exibard, Emmanuel Filiot and Isma\"el Jecker}

\subjclass{,}

\authorrunning{L.\,Exibard, E.\,Filiot and I.\,Jecker}

\acknowledgements{We warmly thank the anonymous reviewers for their helpful comments and Christof Löding for pointing us to some related references.}

\EventEditors{Igor Potapov, Paul Spirakis, and James Worrell}
\EventNoEds{3}
\EventLongTitle{43rd International Symposium on Mathematical Foundations of Computer Science (MFCS 2018)}
\EventShortTitle{MFCS 2018}
\EventAcronym{MFCS}
\EventYear{2018}
\EventDate{August 27--31, 2018}
\EventLocation{Liverpool, GB}
\EventLogo{}
\SeriesVolume{117}
\ArticleNo{46}
\nolinenumbers 

\begin{document}

\maketitle

\begin{abstract}
    The transducer synthesis problem on finite words asks, given a
    specification $S \subseteq I \times O$, where $I$ and $O$ are sets of
    finite words, whether there exists 
    an implementation $f: I \rightarrow O$ which (1) fulfils the specification, i.e., $(i,f(i))\in S$
    for all $i\in I$, and (2) can be defined by some
    input-deterministic (aka sequential) transducer $\mathcal{T}_f$.
    If such an implementation $f$ exists, the procedure should also output $\mathcal{T}_f$.
    The realisability problem is the corresponding decision problem.\par \noindent
    For specifications given by synchronous
    transducers (which read and write alternately one symbol), this is
    the finite variant of the classical synthesis problem on
    $\omega$-words, solved by B\"uchi and Landweber in 1969,
    and the realisability problem is known to be \textsf{ExpTime-c}
    in both finite and $\omega$-word settings. 
    For specifications given by asynchronous
    transducers (which can write a batch of symbols, or none, in a single step),
    the realisability problem is known to be undecidable.\par \noindent
    We consider here the class of multi-sequential
    specifications, defined as finite unions of sequential
    transducers over possibly incomparable domains.
    We provide optimal decision procedures for the realisability problem
    in both the synchronous and asynchronous setting,
    showing that it is \textsf{PSpace-c}.
    Moreover, whenever the specification is realisable,
    we expose the construction of a sequential transducer that realises it
    and has a size that is doubly exponential, which we prove to be optimal.
\end{abstract}

\section{Introduction}

\subparagraph{The realisability and synthesis problem} In general, the
realisability problem is given by some input and output data domains
$D_\inp, D_\outp$, a \emph{specification} $S \subseteq D_\inp\times D_\outp$
defining for every $d\in D_\inp$ the set of allowed outputs $\{ d'\in D_\outp\mid
(d,d')\in S\}$ (assumed to be non-empty) and a class
of target \emph{implementations} $\mathcal{I}$ consisting of (total) functions
$D_\inp\rightarrow D_\outp$. It asks whether there exists $f\in\mathcal{I}$ such
that for all $d\in D_\inp$, $(d,f(d))\in S$, i.e., the implementation $f$
satisfies the specification. The synthesis problem asks to generate (a
representation of) $f$. So, instead of
designing $f$ and verifying its correctness \emph{a posteriori}, a synthesis algorithm automatically generates $f$
from it, making it \emph{correct by construction}. The
underlying idea behind synthesis is that the
specification may be written in a high-level language, e.g. a logic, and an implementation is a low-level computational
model e.g. an automaton. It is based on the assumption that it is
less error-prone to design a specification, i.e. to describe
\emph{what} a system has to do, than designing the system
itself, i.e. describing \emph{how} it must do it.

\subparagraph{Synchronous transducer synthesis} In the original setting defined
by Church \cite{Chur62,Thomas08}, $D_\inp,D_\outp$ are sets of $\omega$-words over
two alphabets $\Sigma_\inp, \Sigma_\outp$ respectively, and the
specification $S$ is given by an $\omega$-language $L\subseteq
(\Sigma_\inp{\times} \Sigma_\outp)^\omega$ as follows: $S = \{ (\pi_1(w),\pi_2(w))\mid w\in
L\}$, where $\pi_i$ is the projection on the $i$th component. The language $L$ is
represented by an MSO-sentence or, equivalently, an automaton. Such automata are also
called (non-deterministic) \emph{synchronous transducers}, as they can be seen as machines
alternately reading one input symbol and synchronously producing one output
symbol. In Church's setting, the target implementations are synchronous
\emph{sequential} transducers (also called \emph{input-deterministic}): they alternately read one input 
symbol and deterministically produce a symbol to output. Determinism
is required because implementations are required to use only
finite-memory. The
Church's instance of the realisability problem is decidable if the
specification is given in MSO and
\textsf{ExpTime-c} if it is given by a synchronous transducer \cite{BuLa69}.
For LTL specifications, it is 
\textsf{2ExpTime-c} \cite{PnuRos:89}. Motivated by reactive systems,
the synthesis problem from LTL specifications has been revisited recently with
efficient symbolic methods~\cite{Jobstm07c,springerlink:10.1007/978-3-540-75596-8_33,FiliotJR11,conf/cav/Ehlers10,DBLP:journals/sttt/JacobsBBEHKPRRS17}. The
synthesis problem in general has also motivated an active research on
infinite games~\cite{DBLP:journals/jacm/AlurHK02,DBLP:conf/concur/AlfaroHM01,DBLP:journals/iandc/ChatterjeeHP10}.

\subparagraph{Asynchronous transducer synthesis} In the asynchronous
setting, specifications may not strictly alternate between input and output
symbols, hence they can no longer be seen as languages over
$\Sigma_\inp\times \Sigma_\outp$. Similarly, the target implementations may
not be synchronous: the system can delay its production of outputs,
or produce several symbols at once. Transducers, in contrast to 
synchronous transducers, are by definition asynchronous: their
transitions are labelled by pairs $(i,w)$ where $i\in\Sigma_\inp$ is a
symbol and $w\in\Sigma_\outp^*$ a word, possibly empty. Since they are generally non-deterministic, to a
single input word may correspond several output words, and thus
transducers define subsets of $\Sigma_\inp^*\times
\Sigma_\outp^*$. Therefore, they are well-suited to represent
(asynchronous) specifications, and in their sequential version,  asynchronous implementations.  
Any asynchronous specification is realisable by some unambiguous
(functional) asynchronous transducer~\cite{journals/iandc/Kobayashi69,Eilenberg:1974:ALM:540337,berstel2009}. However, evaluating  
unambiguous transducers on arbitrarily long or even infinite input words
may require arbitrarily large memory. Therefore, just as in
Church's setting, a sequentiality requirement can be put on
implementations for efficient memory usage. However, the realisability of asynchronous specifications by (asynchronous)
sequential transducers, which is called the \emph{sequential
  uniformisation problem} in transducer-theoretic terms, is
undecidable for finite words~\cite{CarayolL14,DBLP:conf/icalp/FiliotJLW16}. If the
specification is finite-valued (i.e. any input word has a constant
number of output words), the problem is in
\textsf{3ExpTime}~\cite{DBLP:conf/icalp/FiliotJLW16}. The proof
of~\cite{DBLP:conf/icalp/FiliotJLW16} is based on Ramsey's theorem and word 
combinatorics arguments, and it is not clear how to reduce the 
complexity. This raises the question of whether there are natural and
non-trivial subclasses of asynchronous specifications with better
complexity. 


\subparagraph{Multi-sequential specifications} In this
paper, we consider a class of specifications $S$ on finite words,
i.e. $S\subseteq \Sigma_\inp^*\times \Sigma_\outp^*$, which strictly restricts
the class of finite-valued specifications to so-called
multi-sequential specifications. Such class
is obtained by closure under finite unions of graphs
of sequential functions. Precisely, $S = \bigcup_{i=1}^n S_i$ where $S_i$ is the
graph of a (partial) function $f_i : \Sigma_\inp^*\rightarrow
\Sigma_\outp^*$ defined by a sequential  transducer. Likewise, a transducer is
multi-sequential if it is a union of state-disjoint sequential transducers. For instance, consider the specification $S$ consisting of the pairs
$(w,w')$ such that $w'$ is a subword of $w$ of fixed length $k$. This
specification is multi-sequential: $S = \bigcup_{w' \in \Sigma_\inp^k}
S_{w'}$ where $S_{w'} = \{ (w,w')\mid w'\text{ subword of }
w\}$. Clearly, $S_{w'}$ can be represented by a sequential transducer,
because $w'$ is fixed: once the first symbol of $w'$ is met in $w$,
output it, and proceed to the second symbol of $w'$, etc., until the
last symbol of $w'$ is produced, reject otherwise.
The notion of multi-sequentiality has
been introduced for functions in \cite{DBLP:conf/stacs/ChoffrutS86} and studied for relations in
\cite{DBLP:conf/dlt/JeckerF15}.  An important property of
the class of multi-sequential specifications is its decidability in
{\sf PTime}: Given a transducer, it is decidable in \textsf{PTime}
whether it defines a multi-sequential
specification~\cite{DBLP:conf/dlt/JeckerF15}. This fact and their
natural definition as closure of sequential functions under union make
multi-sequential specifications a good candidate for a class of
specifications with better complexity than the known results of the
literature. 


\subparagraph{Contributions} We investigate the complexity of sequential
transducer synthesis from specifications defined by multi-sequential transducers
on finite words. We show that both in the synchronous and asynchronous settings,
the realisability problem is \textsf{PSpace-complete}. To the best of our
knowledge, it is the first non-trivial class of specifications which admits a
realisability test below \textsf{ExpTime}. If the specification is realisable,
we show how to extract an implementation as a winning strategy in a two-player
game called the \emph{synthesis game}. It is parameterised by a value
$k\in\mathbb{N}$ which bounds the maximal number of output symbols which can be
queued before being outputted, allowing for an incremental synthesis algorithm.
To keep track of such output symbols, we use the notion of
\emph{delay}~\cite{BealCPS03}.

\subparagraph{Difficulties and examples} 
Let us briefly explain what are the main difficulties to 
overcome. Consider $S = \bigcup_i S_i$ a multi-sequential
specification. If all of the $S_i$ have disjoint domains, then $S$ is
a function, which is realisable by a sequential transducer iff it is
sequential. The latter can be tested in {\sf
  PTime}~\cite{BealCPS03}. The problem becomes more interesting and
challenging when the $S_i$ have domains which are not necessarily
disjoint. Consider for example the 2-sequential transducer
$\mathcal{D}_1\cup \mathcal{D}_2$ of Fig.~\ref{fig:ex}.
The transducer $\mathcal{D}_1$ accepts the words containing at least two $a$'s, and replaces $b$'s with $a$'s,
and $\mathcal{D}_2$ accepts the words containing at least one $b$, and replaces $a$'s with $b$'s.
This specification can be realised by a sequential
transducer which waits two steps before outputting something, since it then knows whether the input contains at least one $b$
or two $a$'s. It then behaves as
$\mathcal{D}_1$ in the first case, and as
$\mathcal{D}_2$ in the second.

This example shows that a sequential realiser may have to wait before
reacting, keeping in memory what remains to be output in the
future.
Take on the contrary the $2$-sequential transducer 
$\mathcal{D}'_1\cup \mathcal{D}'_2$ of Fig.~\ref{fig:ex}, which is the same
as $\mathcal{D}_1\cup \mathcal{D}_2$ except that it can additionally
read and copy $c$'s at any moment. In that
case, a sequential realiser would have to store arbitrary long
sequences of $c$'s before outputting them, for instance when processing
words in $ac^*\{a,b,c\}^*$. In particular, this specification is not
sequentially realisable. 

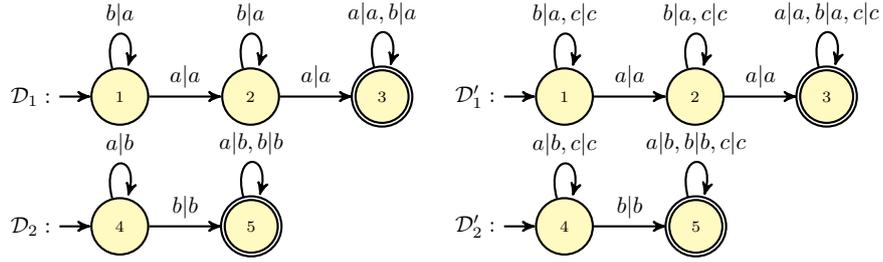
\begin{figure}[ht]
\centering

\begin{tikzpicture}[->,>=stealth',auto,node distance=0.95cm,thick,scale=0.9,every node/.style={scale=0.85}]
\tikzstyle{every state}=[fill=yellow!30,text=black, font=\scriptsize]

\node[state, initial, initial text={$\mathcal{D}_1:$}]			(i)	{1};
\node[state, right= of i]								(p)	{2};
\node[state, right= of p,accepting]						(q)	{3};

\path (i) edge				node[above]	{$a|a$}	(p); 
\path (p) edge				node[above]	{$a|a$}	(q); 
\path (i) edge[loop above]		node[above]	{$b|a$}	(i); 
\path (p) edge[loop above]		node[above]	{$b|a$}	(p); 
\path (q) edge[loop above]		node[above]	{$a|a,b|a$}	(q);

\node[state, initial, initial text={$\mathcal{D}_2:$}, below = of i]	(i1)	{4};
\node[state, right= of i1,accepting]						(p1)	{5};

\path (i1) edge				node[above]	{$b|b$}	(p1); 
\path (i1) edge[loop above]	node[above]	{$a|b$}	(i1); 
\path (p1) edge[loop above]	node[above]	{$a|b,b|b$}	(p1);

\node[state, initial, initial text={$\mathcal{D}_1':$}]	at(6.5,0)	(i)	{1};
\node[state, right= of i]								(p)	{2};
\node[state, right= of p,accepting]						(q)	{3};

\path (i) edge				node[above]	{$a|a$}	(p); 
\path (p) edge				node[above]	{$a|a$}	(q); 
\path (i) edge[loop above]		node[above]	{$b|a,c|c$}	(i); 
\path (p) edge[loop above]		node[above]	{$b|a,c|c$}	(p); 
\path (q) edge[loop above]		node[above]	{$a|a,b|a,c|c$}	(q);

\node[state, initial, initial text={$\mathcal{D}_2':$}, below = of i]	(i1)	{4};
\node[state, right = of i1,accepting]						(p1)	{5};

\path (i1) edge				node[above]	{$b|b$}	(p1); 
\path (i1) edge[loop above]	node[above]	{$a|b,c|c$}	(i1); 
\path (p1) edge[loop above]	node[above]	{$a|b,b|b,c|c$}	(p1); 

\end{tikzpicture}

\caption{Four synchronous sequential transducers.}
\label{fig:ex}
\end{figure}


\subparagraph{Structure of the paper}
After a formal definition of transducer synthesis (Section~\ref{sec:prelims}), we solve the
synchronous case and provide a characterisation of realisable
synchronous multi-sequential specifications, decidable in
\textsf{PSpace} (Section~\ref{sec:synchronous}). Then, we present the
notion of synthesis game (Section~\ref{sec:delay}), which is
a useful tool for the proofs and also to get a synthesis
procedure. For the asynchronous setting, we define a recursive characterisation of realisable multi-sequential
specifications and show that it can be decided in
\textsf{PSpace} (Section~\ref{sec:async}).

\subparagraph{Related work}
Games with delays have been used in~\cite{CarayolL14,DBLP:conf/csl/FridmanLZ11}. Perhaps the closest formulation to ours is that of~\cite{CarayolL14}. However, it is tailored to automatic relation. Our game structure is more general, as it is defined for uniformising any transducer (defining a rational relation). In particular, our game structure is exponentially larger that the one of~\cite{CarayolL14}.

We would also like to mention an interesting related line of
works on $\omega$-words, where the specification is synchronous, but the
implementation may be asynchronous
\cite{DBLP:conf/csl/FridmanLZ11,conf/fossacs/HoltmannKT10,DBLP:journals/corr/KleinZ14,DBLP:conf/fsttcs/KleinZ16,DBLP:journals/ita/Zimmermann16,abs-1709-03539,Zimmermann17}. 
Unlike the setting where the specification and implementations
are both asynchronous, the realisability problem is decidable here,
for $\omega$-regular specifications (i.e., regular
$\omega$-languages over $\Sigma_\inp\times \Sigma_\outp$), and 
\textsf{ExpTime-c} if the specification is given by a parity automaton
\cite{DBLP:journals/corr/KleinZ14}. In this
setting, the authors often consider a notion of delay games.
In these games, the delay is a quantitative notion,
corresponding to the waiting time before outputting a symbol, while
for us, a delay is a word that still remains to be output (this is a
standard terminology in transducer theory). It is known in particular that
constant ``waiting time'' (depending on the specification) is always
sufficient to win, for $\omega$-regular specifications. This
is different to our setting:
for instance, the function $f$ mapping any word of the form
$a^n\sigma$, for $n>0$ and $\sigma\in\{a,b\}$, to $\sigma$ is
realisable by a sequential transducer, but the production of $a$ and $b$
might have to be delayed for an unbounded amount of time.

\section{Transducer synthesis problem}\label{sec:prelims}

\subparagraph{Words} 
For an alphabet $\Sigma$, we
denote by $\Sigma^*$ the set of finite words over it, and by
$\epsilon$ the empty word. The length $|w|$ of a word $w$ is its
number of symbols. For $k\in\mathbb{N}$, we denote by $\Sigma^k$
(resp. $\Sigma^{\leq k}$) the set of words of length $k$ (resp. at most
$k$). For $u,v\in\Sigma^*$, we
write $u\preceq v$ if $u$ is a prefix of $v$, and denote by $u^{-1}v$
the word such that $u(u^{-1}v)=v$. For $L\subseteq \Sigma^*$, the residual language $u^{-1}L$ is 
$u^{-1}L = \{ u'\mid uu'\in L\}$. Given $S\subseteq
\Sigma^*{\times} \Gamma^*$, and $(u,v)\in \Sigma^*{\times}\Gamma^*$, 
the residual relation
$(u,v)^{-1}S$ is defined by $(u,v)^{-1}S = \{
(u',v')\mid (uu',vv')\in S\}$.

\subparagraph{Automata} 
In this paper, finite (non-deterministic) automata
over an alphabet $\Sigma$ are denoted as
tuples $\aut = (\Sigma,Q,I,F,\Delta)$ where $\Sigma$ is the
alphabet, $Q$ the set of states, among which $I$ (resp. $F$) denotes
the initial (resp. final or accepting) states, and $\Delta \subseteq
Q\times \Sigma\times Q$ is the transition relation. $\aut$ is
deterministic if there is only one initial state and for all $(q,\sigma)\in Q\times \Sigma$,
there exists at most one $q' \in Q$ such that $(q,\sigma,q')\in\Delta$.

A \emph{run} of $\aut$ on a word $w=\sigma_1\dots \sigma_n$ consists in either a
single state $q\in Q$ if $n = 0$, or a sequence $r\in \Delta^*$ of $n$
transitions $t_1\dots t_n$ such that the target state of $t_i$ equals
the source state of $t_{i+1}$ for all $1\leq i<n$. It is said to be
initial if the source state of $t_1$ is initial, and accepting if the
target state of $t_n$ is accepting. If $p$ is the
source state of $t_1$ and $q$ the target state of $t_n$, we may write 
$p\myxrightarrow{w}_\aut q$ to mean that there exists a run from $p$ to
$q$ on $w$. The language accepted by an automaton $\aut$, denoted $L(\aut)$, 
is the set of words admitting an accepting run.
A state $q\in Q$ is
reachable (resp. co-reachable) if there is a run from an initial state
(resp. to a final state) for some $u\in\Sigma^*$.
A state is said to be \emph{useful} if it is both reachable and co-reachable, and
$\aut$ is said to be \emph{trim} if all its states are useful. It is
well-known that any automaton can be transformed into an equivalent
trim automaton in \textsf{PTime}. Given two automata $\aut_1 = (\Sigma,Q_1,I_1,F_1,\Delta_1)$ and
$\aut_2 = (\Sigma,Q_2,I_2,F_2,\Delta_2)$, their
disjoint union $\aut_1\uplus \aut_2$ is the automaton $(\Sigma,
Q_1\uplus Q_2, I_1\uplus I_2,F_1\uplus F_2,\Delta_1\uplus
\Delta_2)$.

\subparagraph{Transducers} A \emph{transducer}\footnote{Our
  definition is sometimes called \emph{real-time} transducer in the
  literature, in contrast to transducers with $\epsilon$-input
  transitions. For the purpose of this paper, this does not make a difference.} over
two alphabets $\Sigma,\Gamma$
is a tuple $\tra = (\aut,\rho,\tau)$ such that $\aut =
(\Sigma,Q,I,F,\Delta)$ is an automaton over $\Sigma$, called the
\emph{input automaton}, $\rho :
\Delta\rightarrow \Gamma^*$ is a mapping, called the output function, associating with every transition an
output word, and $\tau : F\rightarrow \Gamma^*$ is a
terminal function associating with every accepting state an output
word. Given a run $r$ of $\aut$ on a word $w$,
its output $\out(r)\in\Gamma^*$ is defined by $\epsilon$ if $w=\epsilon$, and by 
$\rho(t_1)\dots\rho(t_n)$ if $r = t_1\dots t_n$ for some $n\geq
1$. We write $p\myxrightarrow{u|v}_\tra q$ whenever there exists a run $r$
of $\aut$ on $u$ from $p$ to $q$, such that $v = \out(r)$, and say
that $r$ \emph{produces} $v$. The \emph{relation} defined by $\tra$ is the set $\inter{\tra}$ of pairs 
$(u,v\tau(q))\in\Sigma^*\times \Gamma^*$ such that $p\myxrightarrow{u|v}_\tra
q$ for $p\in I$ and $q\in F$. We define $\dom(\tra)$ by  $\dom(\tra) =
\dom(\inter{\tra}) = L(\aut)$. 

A transducer is \emph{trim} if its input automaton is trim. It is called \emph{sequential}
if its input automaton is deterministic, and
\emph{functional} if $\inter{\tra}$ is a function, i.e. for all
$u\in\dom(\tra)$, there exists at most one pair $(u,v)\in \inter{\tra}$. In
that case we let $\tra(u) = v$. Note that any sequential
transducer is functional. A transducer $\tra = (\aut,\rho,\tau)$ is called
\emph{synchronous} (or sometimes \emph{letter-to-letter} in the
literature) if, whenever it reads an input symbol, it produces
exactly one output symbol, i.e. for all transition $t$, $|\rho(t)|=1$,
and $\tau(q) = \epsilon$ for all accepting state $q$. For example, consider the transducer $\mathcal{D}_1$ on
Fig.~\ref{fig:ex} (the terminal function is assumed to output
$\epsilon$ and is not depicted).  It is sequential and synchronous. Its domain is 
$L = b^*ab^*a(a+b)^*$.

Two transducers are said to be
\emph{equivalent} if they define the same relation. 
Finally, the disjoint union of transducers is
naturally defined as the disjoint union of their input automata and
the disjoint union of their output functions (seen as graphs). For all
transducers $\tra_1,\tra_2$, we have $\inter{\tra_1\uplus
\tra_2} = \inter{\tra_1}\cup \inter{\tra_2}$.

\subparagraph{Transducer Synthesis Problem} Let
$\Sigma_\inp,\Sigma_\outp$ be two alphabets of input and output
symbols respectively. They may not necessarily be disjoint. A 
\emph{specification} is a subset of $\Sigma_\inp^*\times
\Sigma_\outp^*$, and an \emph{implementation} is a
function, possibly partial, from $\Sigma_\inp^*$ to $\Sigma_\outp^*$. 
The transducer \emph{realisability} problem asks, given a
specification $S$ defined by a transducer $\tra$, i.e. $S =
\inter{\tra}$, whether there exists a sequential transducer
$\mathcal{I}$ such that $(1)$ $\dom(\mathcal{I}) = \dom(\tra)$ and
$(2)$ for all $u\in \dom(\tra)$, $(u,\mathcal{I}(u))\in \inter{\tra}$.
In that case, we say that $\mathcal{I}$ realises $S$ (or $\tra$), and
that $S$ is realisable by a sequential transducer, or sequentially
realisable. We also say that $\mathcal{I}$ is a \emph{realiser} of
$S$. The synthesis problem asks to output $\mathcal{I}$. 
The realisability problem is undecidable in general~\cite{CarayolL14,DBLP:conf/icalp/FiliotJLW16}, but
decidable, in \textsf{3ExpTime}, if $\tra$ is finite-valued, i.e. there exists
$k\in\mathbb{N}$ such that for all $u\in \dom(\tra)$, 
$|\{ v\mid (u,v)\in\inter{\tra}\}|\leq k$ \cite{DBLP:conf/icalp/FiliotJLW16}. 


\subparagraph{Multi-sequential specifications} A transducer $\tra$ is called
\emph{$k$-sequential} if it is the disjoint union of $k$ sequential
transducers. It is called \emph{multi-sequential} if it is
$k$-sequential for some $k$. Observe that when the $k$ sequential
transducers have pairwise disjoint domains, then $\tra$ is functional,
but it may not be the case in general.
Deciding whether given a transducer $\tra$, there exists an equivalent 
multi-sequential transducer $\tra'$, can be done in \textsf{PTime}; however, $\tra'$ may be
exponentially larger than $\tra$ \cite{DBLP:conf/dlt/JeckerF15}. Minimising the number of
sequential transducers of the disjoint union is also doable: 
deciding whether $\tra$ is
equivalent to some $k$-sequential transducer for $k$ given in unary
is decidable in \textsf{PSpace} \cite{DBLP:conf/fossacs/DaviaudJRV17}.
In this paper, we consider \emph{multi-sequential specification},
i.e. relations $S \subseteq \Sigma_\inp^*\times \Sigma_\outp^*$
defined by multi-sequential transducers.

\subparagraph{\textsf{PSpace}-hardness}
\label{sec:PSpaceHardness}
In both the synchronous and asynchronous case, the realisability problem of
multi-sequential specifications by (a)synchronous sequential transducers is
\textsf{PSpace}-hard.

    We build a reduction from the emptiness problem of the intersection of $n$
DFA $\aut_1,\dots,\aut_n$ on some alphabet $\Sigma$, proven \textsf{PSpace-c}
in~\cite{conf/focs/Kozen77}. We define a specification $S$ over $\Sigma\cup
\{\#,a,b\}$ by $S = \bigcup_{i=1}^n (S_i\cup N_i)$ where $S_i = \{ (w \#^m
\sigma, w \sigma \#^m)\mid \sigma\in\{a,b\},m\geq 0, w\in L(\aut_i)\}$ and $N_i
= \{ (w \#^m \sigma, w \#^m \sigma)\mid \sigma\in\{a,b\}, m\geq 0, w\notin
L(\aut_i)\}$. If there exists $w\in \bigcap_{i=1}^n L(\aut_i)$, then on the
domain $w\#^*\{a,b\}$, the specification is a function which is not definable by
any sequential transducer, thus not sequentially realisable, since it would
imply counting the $\#$s (in the synchronous setting, it suffices to take $m=1$
since a synchronous transducer would be forced to guess the future). Conversely,
if $\bigcap_{i=1}^n L(\aut_i)=\varnothing$, then the identity function
(trivially definable by a synchronous sequential transducer) realises the
specification.

It is readily seen that each $S_i$ (resp. $N_i$) is definable by a 2-(resp.
1-)sequential transducer, hence $S$ is multi-sequential, concluding the proof.

\section{The synchronous setting}\label{sec:synchronous}

In this section, we consider first the synchronous setting, where the
specification is given as a disjoint union of synchronous 
sequential transducers, and the target implementations are synchronous
sequential transducers. Not only is this setting interesting in
itself, but it helps to understand the asynchronous setting. 
First, we characterise the realisable
specifications through a property called the \emph{residual
  property}, then we show it is decidable in \textsf{PSpace}. 

\subparagraph{Residual property} Let $\tra = \biguplus_{i=1}^n \mathcal{D}_i$ be an $n$-sequential transducer on
$\Sigma_\inp,\Sigma_\outp$. Intuitively, the residual property says
that if on some input prefix $u$, two sequential transducers of the
union disagree
on their outputs, i.e. produce different outputs, then a synchronous realiser
of $\inter{\tra}$ must ``drop'' one of the two transducers. However, it
must do so while preserving the residual domain $u^{-1}\dom(\tra)$,
i.e., the realiser must still accept any word of
$u^{-1}\dom(\tra)$. For example, consider again Fig.~\ref{fig:ex} 
and the specification defined by $\mathcal{D}_1\uplus
\mathcal{D}_2$. On input $a$, the two transducers disagree, hence, since we want a synchronous
realiser, a choice has to be made and therefore one of the two
transducers must be dropped. However, by doing so, the residual domain
will not be fully covered by the remaining transducer. For example, if
a realiser chooses to output $a$ when reading $a$, the residual language $b^*$ is not covered anymore. As a
matter of fact, $\mathcal{D}_1\uplus \mathcal{D}_2$ is not realisable
by any sequential and synchronous transducer.

Formally, let $u\in \Sigma_\inp^*$ and let $r_i,r_j$ be runs of some
$\mathcal{D}_i,\mathcal{D}_j$ respectively, on $u$. We say that $r_i$
and $r_j$ agree on their output if $\out(r_i) = \out(r_j)$.  Now, $u$ is called \emph{smooth} if every $\mathcal{D}_i$ admits 
an initial run on input $u$, and all these runs agree on the corresponding output.
The word $u$ is called \emph{critical} if it is not smooth.

We say that $\tra$ satisfies the \emph{residual property} if for every critical prefix $u \in \Sigma^*_\inp$ of a word of $\dom(\tra)$,
there exists a subset $P \subsetneq \{ 1, \ldots, n\}$ satisfying:
\begin{enumerate}
\item
All the transducers $\mathcal{D}_i$, $i \in P$, produce the same
output $\pick{u}$ on $u$;
\item
$u^{-1} \dom(\tra) = \bigcup_{i \in P} u^{-1} \dom(\mathcal{D}_i)$;
\item
$\biguplus_{i \in P} (u,\pick{u})^{-1} \inter{\mathcal{D}_i}$ is realisable by a synchronous and sequential transducer.
\end{enumerate}


\begin{restatable}{theorem}{characsync}\label{thm:characsync}
    A specification $S$ defined by a synchronous
    multi-sequential transducer $\tra$ is realisable by a
    synchronous sequential transducer iff $\tra$ satisfies the
    residual property. 
\end{restatable}

\begin{proof}[Sketch]
    If $\inter{\tra}$ is realised by a synchronous sequential transducer $\mathcal{U}$, for every critical prefix $u$,
    let $P$ be the set of $i$ such that $\mathcal{D}_i$ and $\mathcal{U}$ map the same output to $u$. Property 1
    is satisfied by definition, and the other two follow from the fact that $\mathcal{U}$ is sequential and realises $\inter{\tra}$.

    Conversely, if the residual property is satisfied, we can
    construct a synchronous and sequential realiser. The idea is to
    make a synchronised product of all the transducers
    $\mathcal{D}_i$, and, whenever on some input symbol $\sigma$ at
    least two of them disagree on the output, we know by the residual property 
    that there exists a subset $P$ of them having the good properties
    $1,2,3$. Then, the realiser just goes on simulating all the
    transducers $\mathcal{D}_i$ corresponding to  $P$ in parallel.

    It also shows that if the property is satisfied, then we can 
    synthesise a realiser, which might however be exponentially larger than
    $\tra$.
\end{proof}

\begin{theorem}\label{thm:complexitysync}
The realisability problem of synchronous multi-sequential
specifications by synchronous sequential transducers is
\textsf{PSpace}-complete. 
\end{theorem}
\begin{proof}[Sketch]
  The \textsf{PSpace}-hardness is obtained by reducing the problem
  from the emptiness problem of the intersection of $n$ DFAs
  (\emph{cf} Section~\ref{sec:prelims} p.~\pageref{sec:PSpaceHardness}).

    To show membership to \textsf{PSpace}, given a transducer $\tra =
    \biguplus_{i=1}^n \mathcal{D}_i$, we show that the residual
    property can be tested by a non-deterministic algorithm running in
    polynomial space. First, we bound the size of witnesses of the
    negation of the property:
    roughly, if there is such witness, namely a critical prefix $u$, then there
    exists a critical prefix $v$ of exponential length (in $\tra$) such that for
    any subset $P\subsetneq \{1,\dots,n\}$, one of the conditions
    $1,2,3$ is falsified. Then, the algorithm guesses the prefix
    $v$ on the fly, simulating all transducers $\mathcal{D}_i$ in
    parallel and keeping their states in memory
    (it also needs a counter for the length of $v$). As soon as
    the transducers disagree on an output symbol, for each subset
    $P\subsetneq \{1,\dots,n\}$ (they can obviously be enumerated
    using only polynomial space), the algorithm checks whether 
    property $1$, $2$ or $3$ is falsified. Checking property $1$
    is easy: it suffices to look at the symbols produced when reading
    the last input symbol. Checking property $2$ can be done using the
    current set of states reached by the transducers on input $v$, and
    by using any \textsf{PSpace} algorithm for automata
    inclusion. Finally, to check property $3$, it suffices to
    recursively apply the \textsf{PSpace} algorithm described so far
    on a smaller set of transducers. The stack of recursive calls is
    linear in $n$, hence the memory used by the whole procedure  
    remains polynomial.
\end{proof}

\section{The synthesis game}\label{sec:delay}

We now define a 2-player safety game from a transducer $\tra$ such
that if Eve wins the game then $\tra$ is realisable by a sequential
transducer. This game notion will prove useful to show the
correctness of the characterisation of Theorem~\ref{thm:charac},
and may also be used as a practical way
to synthesise implementations, as winning strategies of this game. 
In the asynchronous setting, two different runs of a transducer on
the same input word may not only produce different outputs, but also the same
output at different rates (i.e. one run is ahead, output-wise, of the
other for some time). This leads us to the notion of delays, a classical
tool to compare outputs in transducer theory. Let us define this
notion formally. 

\subparagraph{Delays}
Given two words $u_1,u_2\in\Sigma^*$, their longest common prefix $\ell$ is denoted by
$u_1\wedge u_2$. The \emph{delay} between $u_1$ and $u_2$ is an element of
$\Sigma^*\times \Sigma^*$ defined by $\del(u_1,u_2) = (\ell^{-1}u_1,\ell^{-1}u_2)$.
Intuitively, if a transducer produces
$u_1$ and another one produces $u_2$, then $u_1\wedge u_2$ is what can safely be
output by the two transducers and $\del(u_1,u_2)$ what remains to be produced by
each of them respectively. This notion is naturally extended to tuples of words:
$\del(u_1,\dots, u_n) = (\ell^{-1}u_1,\dots,\ell^{-1}u_n)$ where $\ell =
\bigwedge_{i=1}^n u_i$.

We now introduce notations that are useful when comparing the outputs
of different runs on the same input of a transducer $\tra =
(\aut,\rho, \tau)$ over $\Sigma_\inp,\Sigma_\outp$ with $\aut = (\Sigma_\inp,Q,I,F,\Delta)$. 
Given a pair $(q,w)\in Q\times
\Sigma_\outp^*$, where $w$ is intended to be some delay associated with
state $q$, given a transition $t = (q,\sigma,q')\in\Delta$ and some
output word $u$ prefix of $w\rho(t)$, we denote by $next((q,w),t,u)$ the
``next'' pair (state,delay), assuming that $u$ is output, i.e. 
$next((q,w),t,u) = (q',u^{-1}w\rho(t))$.
More generally, given a (total) function $D : Q \rightarrow 2^{\Sigma_\outp^*}$
associating each state with a set of delays, we let $live(D)  = \{
q\in Q\mid D(q)\neq\varnothing\}$. For $\sigma\in\Sigma_\inp$,
$next(D,\sigma)$ maps every state which can be reached from $\dom(D)$ by reading $\sigma$
to the corresponding delays obtained by outputting the longest common prefix of the
words that can be formed from the previous delays and the output on
these transitions. Formally, we call \emph{safe output of $D$ for $\sigma$} the
word $\ell = \bigwedge \{w\rho(t)\mid q \in live(D),
w \in D(q), t = (q,\sigma,q')\in \Delta\}$. Then
$next(D,\sigma) = \{next((q,w),t,\ell)\mid  q \in live(D),
w \in D(q), t = (q,\sigma,q')\in \Delta\}$.

\subparagraph{The synthesis game}
In the synchronous setting, synthesis problems are classically solved by
reduction to two-player games in which the players alternately choose one input
symbol (the adversary, whom we call Adam) and one output symbol (the
protagonist, called Eve). Their interaction induces a pair of input and output
words by concatenating their respective symbols, and the protagonist wins if
such pair satisfies the specification, or if the input word is out of the
domain. Then, a finite-memory winning strategy in the game corresponds to an
implementation of the specification.

In the asynchronous setting, the protagonist may choose arbitrary
output words at each round instead of a single symbol, and one needs to
introduce output delays in the game in order to define the
winning condition in a regular manner. The game we now present follows
this idea. Given a
transducer $\tra = (\aut, \rho,\tau)$ with $\aut =
(\Sigma_\inp,Q,I,F,\Delta)$, $\rho : \Delta\rightarrow \Sigma_\outp^*$
and $\tau : F\rightarrow \Sigma_\outp^*$, we build a two-player
safety game $G_\tra =
(V_\forall,V_\exists,A_\forall,A_\exists,T_\forall,T_\exists,\textsf{Safe})$,
called the \emph{synthesis game}, whose vertices keep track of the runs in $\tra$
and the associated delays. More precisely, it consists of two disjoint
sets of vertices $V_\forall = 2^Q\times (Q\rightarrow 2^{\Sigma_\outp^*})$  and $V_\exists = V_\forall\times \Sigma_\inp$, respectively
controlled by Adam and Eve. The initial vertex is $v_0 = (I,D_0)\in
V_\forall$ where $D_0(q) = \varnothing$ if $q\not\in I$, and
$D_0(q)=\{\epsilon\}$ otherwise.

Eve's vertices are Adam's vertices extended with the last input symbol picked
by Adam. Suppose now that the game has been played for some rounds and is
currently in some vertex $(C,D)$ of Adam. Along these rounds, Adam has chosen
a sequence $u$ of input symbols, and Eve has chosen a set
of runs
over $u$ from the initial states. $C$ is the set of states in which these runs end.
Each run induces some delays compared to the longest common prefix of all the outputs they can
produce. $D$ maps each state to the delays of the runs ending in it.
Eve's actions consist in selecting some of these runs to
prevent some delays to grow too high, i.e., she can drop from any set
$D(q)$ some of its elements.
By restricting the set of possible runs, Eve
can be in a situation where some state $q$ of $C$ is accepting while
none of the states of $live(D)$ is, in which case she loses,
as none of the runs she has selected accepts the input word chosen
by Adam. Such vertices constitute the set of \emph{unsafe} vertices
she needs to avoid. 

More precisely, the set of Adam's \emph{transitions} $T_\forall$  and
Eve's transitions $T_\exists$ are defined as follows. 
From any game position $(C,D)$, Adam can pick a symbol
$\sigma\in\Sigma_\inp$ and the game evolves to the position
$(C,D,\sigma)$. From $(C,D,\sigma)$, Eve's actions is a subset
$\alpha\subseteq next(D,\sigma)$ (she can ``drop'' some pairs of $next(D,\sigma)$), and the game evolves to 
$(C', D_\alpha)$ where $C'$ is the set of states reached from $C$ by
reading $\sigma$, and $D_\alpha$ maps any $q\in Q$ to the set $\{
w\mid (q,w)\in\alpha\}$. 

Given $K\in\mathbb{N}$, we define the \emph{$K$-synthesis
game} 
as the restriction of $G_\tra$ to delays of length at most $K$:
$G_{\tra,K} = (V_\forall^K,V^K_\exists, v^K_0, A^K_\forall, A^K_\exists,
T^K_\forall, T^K_\exists,\textsf{Safe}^K)$,
where $V_\forall^K = 2^Q\times (Q\rightarrow 2^{\Sigma_\outp^{\leq K}})$,
$A_\exists \subseteq Q\times \Sigma_\outp^{\leq K}$, etc. There
is no deadlock in $G_{\tra,K}$ because Eve can always play
$\varnothing$, at the risk of going to an unsafe position.

\subparagraph{Example}
First, note that by definition of the game, any reachable vertex
$(C,D)$ or $(C,D,\sigma)$ satisfies $live(D)\subseteq C$. Figure~\ref{fig:1delay} represents the $1$-synthesis game for
$\mathcal{D}_1 \cup \mathcal{D}_2$ (\emph{cf} Figure~\ref{fig:ex}). The states
depicted in a vertex correspond to $C$, together with their values by
$D$ (thanks to the previous remark, there is no need to represent the
values $D$ assigns for the states outside $C$). 
The circle vertices are Eve's positions, whose labels are not
depicted, as they are just the label of their predecessor vertex
extended with Adam's action. Bold nodes correspond to the unsafe
states. Let us now explain how the game proceeds in more detail. First, since both $\mathcal{D}_1$ and $\mathcal{D}_2$ are complete and sequential, for each state $(C,D)$ of Adam,
$C$ contains exactly one state of $\mathcal{D}_1$ and one state of $\mathcal{D}_2$.
Eve's actions in the synthesis game, which consist in dropping a subset of pairs (state,delay),
actually correspond here to ``dropping'' one of the two sequential transducers:
at any moment, she can choose to drop $\mathcal{D}_2$, which leads her into the red part of the game,
or to drop $\mathcal{D}_1$, which leads her into the blue part of the game.
Note that once she has dropped one of the transducers, she is stuck in the corresponding part. 

The initial vertex, owned by Adam, corresponds to being in the initial states of both $\mathcal{D}_1$ and $\mathcal{D}_2$, with no delays. 
If Adam chooses to play $a$ as the first input letter, Eve has four choices.
Either she keeps both transducers, with a delay of length $1$, or she
drops one of them, or both (not depicted). 
If Eve chooses to drop $\mathcal{D}_2$, respectively $\mathcal{D}_1$, Adam can then play a $b$, respectively an $a$,
which leads her into an unsafe state.
Note that this proves that Eve cannot win the $0$-synthesis game corresponding to $\mathcal{D}_1 \cup \mathcal{D}_2$.
If Eve keeps both, she has to drop one of them once Adam plays a second letter, since otherwise the delay would grow larger than $1$.
However, in both cases she has a move which ensures her a win: if Adam plays a second $a$, Eve can safely drop $\mathcal{D}_2$
since the accepting state of $\mathcal{D}_1$ has been reached, and if Adam plays $b$,
Eve can safely drop $\mathcal{D}_1$ since the accepting state of
$\mathcal{D}_2$ has been reached. If Adam chooses to play a $b$ in the
first place, Eve can immediately drop $\mathcal{D}_1$ and win. Hence,
Eve wins the $1$-synthesis game associated to $\mathcal{D}_1 \cup
\mathcal{D}_2$. The described strategy then directly induces a sequential transducer
realising the specification.

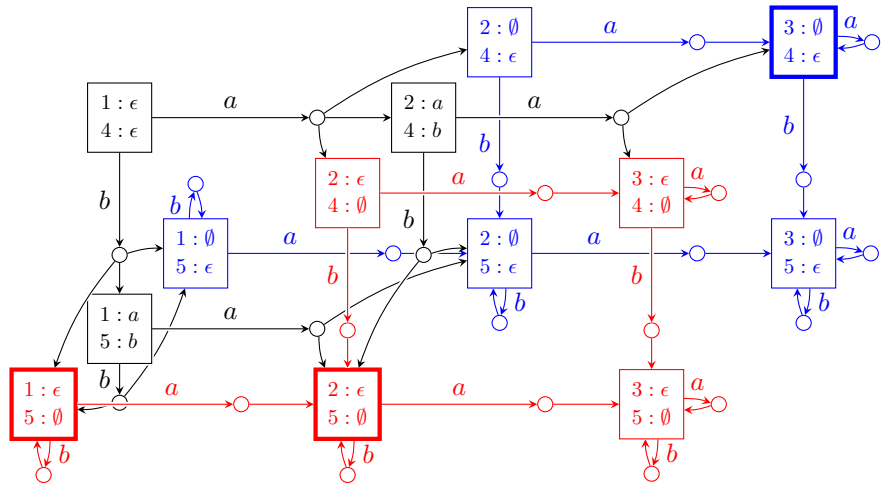
\begin{figure}[ht]
  \centering

\pgfdeclarelayer{T1}
\pgfdeclarelayer{T11}
\pgfdeclarelayer{T2}
\pgfdeclarelayer{T22}
\pgfsetlayers{T2,T22,main,T11,T1}

\begin{tikzpicture}[->, >=stealth,bend angle=10]

\newcommand\aX{2}
\newcommand\bX{0}
\newcommand\aY{1.4}
\newcommand\bY{0}
\newcommand\s{1.6}


\tikzset{
adam/.style    = {draw = white, line width=\s pt, rectangle,minimum width=30pt,minimum height=30pt,scale=.8},
eve/.style     = {draw = white, line width=\s pt, circle,inner sep=2pt},
every path/.style ={-,draw = white, line width=\s pt},
winning/.style = {very thick, font=\boldmath},
unsafe/.style  = {fill = red!70!white}
}

\node[adam]	at(0*\aX + \bX,0*\aY + \bY)		(i)			{$\arraycolsep=1.4pt \begin{array} {lcl}	1 & : & \epsilon	\\	4 & : & \epsilon		\end{array}$};
\node[adam]	at(2*\aX + \bX,0*\aY + \bY)		(a)			{$\arraycolsep=1.4pt \begin{array} {lcl}	2 & : & a		\\	4 & : & b			\end{array}$};
\node[adam]	at(0*\aX + \bX,-2*\aY + \bY)		(b)			{$\arraycolsep=1.4pt \begin{array} {lcl}	1 & : & a		\\	5 & : & b			\end{array}$};

\begin{pgfonlayer}{T22}
\node[eve]		at(1.3*\aX + \bX,0*\aY + \bY)		(i-a)			{};
\node[eve]		at(3.3*\aX + \bX,0*\aY + \bY)		(a-aa)		{};
\node[eve]		at(0*\aX + \bX,-1.3*\aY + \bY)		(i-b)			{};
\node[eve]		at(2*\aX + \bX,-1.3*\aY + \bY)		(a-ab)		{};
\node[eve]		at(1.3*\aX + \bX,-2*\aY + \bY)		(b-ba)		{};
\node[eve]		at(0*\aX + \bX,-2.7*\aY + \bY)		(b-bb)		{};
\end{pgfonlayer}

\path		(i)		edge	[-,draw = white, line width=\s pt]		node[above]	{}	(i-a);
\path		(i-a)		edge[-,draw = white, line width=\s pt]		node[above]	{}	(a);
\path		(a)		edge[-,draw = white, line width=\s pt]		node[above]	{}	(a-aa);
\path		(i)		edge[-,draw = white, line width=\s pt]		node[left]		{}	(i-b);
\path		(i-b)		edge[-,draw = white, line width=\s pt]		node[above]	{}	(b);
\path		(a)		edge[-,draw = white, line width=\s pt]		node[left]		{}	(a-ab);
\path		(b)		edge	[-,draw = white, line width=\s pt]		node[above]	{}	(b-ba);
\path		(b)		edge	[-,draw = white, line width=\s pt]		node[left]		{}	(b-bb);

\path		(i)		edge		node[above]	{}	(i-a);
\path		(i-a)		edge		node[above]	{}		(a);
\path		(a)		edge		node[above]	{}	(a-aa);
\path		(i)		edge		node[left]		{}	(i-b);
\path		(i-b)		edge		node[above]	{}		(b);
\path		(a)		edge		node[left]		{}	(a-ab);
\path		(b)		edge		node[below]	{}	(b-ba);
\path		(b)		edge		node[left]		{}	(b-bb);


\tikzset{
adam/.style    = {draw=black, rectangle,minimum width=30pt,minimum height=30pt,scale=.8,fill = white},
eve/.style     = {draw=black, circle,inner sep=2pt, fill = white},
every path/.style ={black}
}

\node[adam]	at(0*\aX + \bX,0*\aY + \bY)		(i)			{$\arraycolsep=1.4pt \begin{array} {lcl}	1 & : & \epsilon	\\	4 & : & \epsilon		\end{array}$};
\node[adam]	at(2*\aX + \bX,0*\aY + \bY)		(a)			{$\arraycolsep=1.4pt \begin{array} {lcl}	2 & : & a		\\	4 & : & b			\end{array}$};
\node[adam]	at(0*\aX + \bX,-2*\aY + \bY)		(b)			{$\arraycolsep=1.4pt \begin{array} {lcl}	1 & : & a		\\	5 & : & b			\end{array}$};
\node[eve]		at(1.3*\aX + \bX,0*\aY + \bY)		(i-a)			{};
\node[eve]		at(3.3*\aX + \bX,0*\aY + \bY)		(a-aa)		{};
\node[eve]		at(0*\aX + \bX,-1.3*\aY + \bY)		(i-b)			{};
\node[eve]		at(2*\aX + \bX,-1.3*\aY + \bY)		(a-ab)		{};
\node[eve]		at(1.3*\aX + \bX,-2*\aY + \bY)		(b-ba)		{};
\node[eve]		at(0*\aX + \bX,-2.7*\aY + \bY)		(b-bb)		{};

\path		(i)		edge		node[above]	{$a$}	(i-a);
\path		(i-a)		edge		node[above]	{}		(a);
\path		(a)		edge		node[above]	{$a$}	(a-aa);
\path		(i)		edge		node[left]		{$b$}	(i-b);
\path		(i-b)		edge		node[above]	{}		(b);
\path		(a)		edge		node[below left]		{$b$}	(a-ab);
\path		(b)		edge		node[above]	{$a$}	(b-ba);
\path		(b)		edge		node[left]		{$b$}	(b-bb);


\begin{pgfonlayer}{T1}

\tikzset{
adam/.style    = {draw = white, line width=\s pt, rectangle,minimum width=2pt,minimum height=2pt,scale=.8,fill=white},
eve/.style     = {draw = white, line width=\s pt, circle,inner sep=2pt,fill=white},
every path/.style ={-,draw = white, line width=\s pt},
winning/.style = {very thick, font=\boldmath},
unsafe/.style  = {fill = red!70!white}
}

\renewcommand\bX{-1}
\renewcommand\bY{-1}

\node[adam]	at(2*\aX + \bX,0*\aY + \bY)		(a2)			{};
\node[adam]	at(4*\aX + \bX,0*\aY + \bY)		(aa2)			{};
\node[adam]	at(0*\aX + \bX,-2*\aY + \bY)		(b2)			{};
\node[adam]	at(2*\aX + \bX,-2*\aY + \bY)		(ba2)			{};
\node[adam]	at(4*\aX + \bX,-2*\aY + \bY)		(baa2)		{};
\node[eve]		at(3.3*\aX + \bX,0*\aY + \bY)		(a2-aa2)		{};
\node[eve]		at(2*\aX + \bX,-1.3*\aY + \bY)		(a2-ab2)		{};
\node[eve]		at(4*\aX + \bX,-1.3*\aY + \bY)		(aa2-baa2)	{};
\node[eve]		at(1.3*\aX + \bX,-2*\aY + \bY)		(b2-ba2)		{};
\node[eve]		at(3.3*\aX + \bX,-2*\aY + \bY)		(ba2-baa2)	{};

\node[eve, right = 0.35cm of aa2]		(aa2-aaa2)	{};
\node[eve, right = 0.35cm of baa2]		(baa2-baaa2)	{};

\node[eve, below = 0.35cm of b2]		(b2-bb2)	{};
\node[eve, below = 0.35cm of ba2]		(ba2-bba2)	{};
\node[eve, below = 0.35cm of baa2]		(baa2-bbaa2)	{};

\path		(a2)			edge		node[above]	{}	(a2-aa2);
\path		(a2-aa2)		edge		node[above]	{}	(aa2);
\path		(a2)			edge		node[left]		{}	(a2-ab2);
\path		(a2-ab2)		edge		node[above]	{}	(ba2);
\path		(aa2)			edge		node[left]		{}	(aa2-baa2);
\path		(aa2-baa2)	edge		node[above]	{}	(baa2);
\path		(b2)			edge		node[above]	{}	(b2-ba2);
\path		(b2-ba2)		edge		node[above]	{}	(ba2);
\path		(ba2)			edge		node[above]	{}	(ba2-baa2);
\path		(ba2-baa2)	edge		node[above]	{}	(baa2);

\path		(b2)			edge	[bend left]	node[right]	{}	(b2-bb2);
\path		(b2-bb2)		edge	[bend left]	node[left]		{}		(b2);

\path		(ba2)			edge	[bend left]	node[right]	{}	(ba2-bba2);
\path		(ba2-bba2)	edge	[bend left]	node[left]		{}		(ba2);

\path		(baa2)		edge	[bend left]	node[right]	{}	(baa2-bbaa2);
\path		(baa2-bbaa2)	edge	[bend left]	node[left]		{}		(baa2);

\path		(aa2)			edge	[bend left]	node[above]	{}	(aa2-aaa2);
\path		(aa2-aaa2)	edge	[bend left]	node[below]	{}		(aa2);

\path		(baa2)		edge	[bend left]	node[above]	{}	(baa2-baaa2);
\path		(baa2-baaa2)	edge	[bend left]	node[below]	{}		(baa2);

\end{pgfonlayer}


\begin{pgfonlayer}{T1}

\tikzset{color = red,
adam/.style    = {draw=red, rectangle,minimum width=30pt,minimum height=30pt,scale=.8},
eve/.style    = {draw=red, circle,inner sep=2pt},
every path/.style ={red}
}

\renewcommand\bX{-1}
\renewcommand\bY{-1}

\node[adam]	at(2*\aX + \bX,0*\aY + \bY)		(a2)			{$\arraycolsep=1.4pt \begin{array} {lcl}	2 & : & \epsilon	\\	4 & : & \emptyset		\end{array}$};
\node[adam]	at(4*\aX + \bX,0*\aY + \bY)		(aa2)			{$\arraycolsep=1.4pt \begin{array} {lcl}	3 & : & \epsilon \\	4 & : & \emptyset		\end{array}$};
\node[adam,ultra thick]	at(0*\aX + \bX,-2*\aY + \bY)		(b2)			{$\arraycolsep=1.4pt \begin{array} {lcl}	1 & : & \epsilon	\\	5 & : & \emptyset		\end{array}$};
\node[adam,ultra thick]	at(2*\aX + \bX,-2*\aY + \bY)		(ba2)			{$\arraycolsep=1.4pt \begin{array} {lcl}	2 & : & \epsilon	\\	5 & : & \emptyset		\end{array}$};
\node[adam]	at(4*\aX + \bX,-2*\aY + \bY)		(baa2)		{$\arraycolsep=1.4pt \begin{array} {lcl}	3 & : & \epsilon	\\	5 & : & \emptyset		\end{array}$};
\node[eve]		at(3.3*\aX + \bX,0*\aY + \bY)		(a2-aa2)		{};
\node[eve]		at(2*\aX + \bX,-1.3*\aY + \bY)		(a2-ab2)		{};
\node[eve]		at(4*\aX + \bX,-1.3*\aY + \bY)		(aa2-baa2)	{};
\node[eve]		at(1.3*\aX + \bX,-2*\aY + \bY)		(b2-ba2)		{};
\node[eve]		at(3.3*\aX + \bX,-2*\aY + \bY)		(ba2-baa2)	{};

\node[eve, right = 0.35cm of aa2]		(aa2-aaa2)	{};
\node[eve, right = 0.35cm of baa2]		(baa2-baaa2)	{};

\node[eve, below = 0.35cm of b2]		(b2-bb2)	{};
\node[eve, below = 0.35cm of ba2]		(ba2-bba2)	{};
\node[eve, below = 0.35cm of baa2]		(baa2-bbaa2)	{};

\path		(a2)			edge[-,draw = white, line width=\s pt]		node[above]	{}	(a2-aa2);
\path		(a2)			edge[-,draw = white, line width=\s pt]		node[left]		{}	(a2-ab2);
\path		(a2-ab2)		edge[-,draw = white, line width=\s pt]		node[above]	{}	(ba2);
\path		(aa2)			edge[-,draw = white, line width=\s pt]		node[left]		{}	(aa2-baa2);
\path		(b2)			edge[-,draw = white, line width=\s pt]		node[above]	{}	(b2-ba2);

\path		(a2)			edge								node[above]	{$a$}	(a2-aa2);
\path		(a2-aa2)		edge		node[above]		{}		(aa2);
\path		(a2)			edge		node [left]			{$b$}	(a2-ab2);
\path		(a2-ab2)		edge		node[above]		{}		(ba2);
\path		(aa2)			edge		node[left]			{$b$}	(aa2-baa2);
\path		(aa2-baa2)	edge		node[above]		{}		(baa2);
\path		(b2)			edge		node[above right]	{$a$}	(b2-ba2);
\path		(b2-ba2)		edge		node[above]		{}		(ba2);
\path		(ba2)			edge		node[above]		{$a$}	(ba2-baa2);
\path		(ba2-baa2)	edge		node[above]		{}		(baa2);

\path		(b2)			edge	[bend left]	node[right]	{$b$}	(b2-bb2);
\path		(b2-bb2)		edge	[bend left]	node[left]		{}		(b2);

\path		(ba2)			edge	[bend left]	node[right]	{$b$}	(ba2-bba2);
\path		(ba2-bba2)	edge	[bend left]	node[left]		{}		(ba2);

\path		(baa2)		edge	[bend left]	node[right]	{$b$}	(baa2-bbaa2);
\path		(baa2-bbaa2)	edge	[bend left]	node[left]		{}		(baa2);

\path		(aa2)			edge	[bend left]	node[above]	{$a$}	(aa2-aaa2);
\path		(aa2-aaa2)	edge	[bend left]	node[below]	{}		(aa2);

\path		(baa2)		edge	[bend left]	node[above]	{$a$}	(baa2-baaa2);
\path		(baa2-baaa2)	edge	[bend left]	node[below]	{}		(baa2);

\end{pgfonlayer}


\begin{pgfonlayer}{T11}

\tikzset{
every path/.style ={-,draw = white, line width=\s pt,shorten <=5pt}
}

\path		(i-a)			edge[bend right]	node[above]	{}		(a2);
\path		(i-b)			edge[bend right]	node[above]	{}		(b2);
\path		(a-aa)		edge[bend right]	node[above]	{}		(aa2);
\path		(a-ab)		edge[bend right]	node[above]	{}		(ba2);
\path		(b-ba)		edge[bend right]	node[above]	{}		(ba2);
\path		(b-bb)		edge[bend left]		node[above]	{}		(b2);

\tikzset{
every path/.style ={black}
}

\path		(i-a)			edge[bend right]	node[above]	{}		(a2);
\path		(i-b)			edge[bend right]	node[above]	{}		(b2);
\path		(a-aa)		edge[bend right]	node[above]	{}		(aa2);
\path		(a-ab)		edge[bend right]	node[above]	{}		(ba2);
\path		(b-ba)		edge[bend right]	node[above]	{}		(ba2);
\path		(b-bb)		edge[bend left]		node[above]	{}		(b2);

\end{pgfonlayer}


\tikzset{color = blue,
adam/.style    = {draw=blue, rectangle,minimum width=30pt,minimum height=30pt,scale=.8},
eve/.style    = {draw=blue, circle,inner sep=2pt},
every path/.style ={blue}
}

\renewcommand\bX{1}
\renewcommand\bY{1}

\begin{pgfonlayer}{main}

\node[adam]	at(2*\aX + \bX,0*\aY + \bY)		(a3)			{$\arraycolsep=1.4pt \begin{array} {lcl}	2 & : & \emptyset	\\	4 & : & \epsilon		\end{array}$};
\node[adam,ultra thick]	at(4*\aX + \bX,0*\aY + \bY)		(aa3)			{$\arraycolsep=1.4pt \begin{array} {lcl}	3 & : & \emptyset 	\\	4 & : & \epsilon		\end{array}$};
\node[adam]	at(0*\aX + \bX,-2*\aY + \bY)		(b3)			{$\arraycolsep=1.4pt \begin{array} {lcl}	1 & : & \emptyset	\\	5 & : & \epsilon		\end{array}$};
\node[adam]	at(2*\aX + \bX,-2*\aY + \bY)		(ba3)			{$\arraycolsep=1.4pt \begin{array} {lcl}	2 & : & \emptyset	\\	5 & : & \epsilon		\end{array}$};
\node[adam]	at(4*\aX + \bX,-2*\aY + \bY)		(baa3)		{$\arraycolsep=1.4pt \begin{array} {lcl}	3 & : & \emptyset	\\	5 & : & \epsilon		\end{array}$};

\end{pgfonlayer}

\begin{pgfonlayer}{T2}

\node[eve]		at(3.3*\aX + \bX,0*\aY + \bY)		(a3-aa3)		{};
\node[eve]		at(2*\aX + \bX,-1.3*\aY + \bY)		(a3-ab3)		{};
\node[eve]		at(4*\aX + \bX,-1.3*\aY + \bY)		(aa3-baa3)	{};
\node[eve]		at(1.3*\aX + \bX,-2*\aY + \bY)		(b3-ba3)		{};
\node[eve]		at(3.3*\aX + \bX,-2*\aY + \bY)		(ba3-baa3)	{};

\node[eve, right = 0.35cm of aa3]		(aa3-aaa3)	{};
\node[eve, right = 0.35cm of baa3]		(baa3-baaa3)	{};

\node[eve, above = 0.35cm of b3]		(b3-bb3)	{};
\node[eve, below = 0.35cm of ba3]		(ba3-bba3)	{};
\node[eve, below = 0.35cm of baa3]		(baa3-bbaa3)	{};

\path		(a3)			edge		node[above]		{$a$}	(a3-aa3);
\path		(a3-aa3)		edge		node[above]		{}		(aa3);
\path		(a3)			edge		node[below left]			{$b$}	(a3-ab3);
\path		(a3-ab3)		edge		node[above]		{}		(ba3);
\path		(aa3)			edge		node[left]			{$b$}	(aa3-baa3);
\path		(aa3-baa3)	edge		node[above]		{}		(baa3);
\path		(b3)			edge		node[above left]	{$a$}	(b3-ba3);
\path		(b3-ba3)		edge		node[above]		{}		(ba3);
\path		(ba3)			edge		node[above left]	{$a$}	(ba3-baa3);
\path		(ba3-baa3)	edge		node[above]		{}		(baa3);

\path		(b3)			edge	[bend left]	node[left]	{$b$}	(b3-bb3);
\path		(b3-bb3)		edge	[bend left]	node[left]		{}		(b3);

\path		(ba3)			edge	[bend left]	node[right]	{$b$}	(ba3-bba3);
\path		(ba3-bba3)	edge	[bend left]	node[left]		{}		(ba3);

\path		(baa3)		edge	[bend left]	node[right]	{$b$}	(baa3-bbaa3);
\path		(baa3-bbaa3)	edge	[bend left]	node[left]		{}		(baa3);

\path		(aa3)			edge	[bend left]	node[above]	{$a$}	(aa3-aaa3);
\path		(aa3-aaa3)	edge	[bend left]	node[below]	{}		(aa3);

\path		(baa3)		edge	[bend left]	node[above]	{$a$}	(baa3-baaa3);
\path		(baa3-baaa3)	edge	[bend left]	node[below]	{}		(baa3);

\end{pgfonlayer}


\begin{pgfonlayer}{T22}

\tikzset{
every path/.style ={-,draw = white, line width=\s pt}
}

\path		(i-a)			edge	[bend left]		node[above]	{}		(a3);
\path		(i-b)			edge	[bend left]		node[above]	{}		(b3);
\path		(a-aa)		edge	[bend left]		node[above]	{}		(aa3);
\path		(a-ab)		edge	[bend left]		node[above]	{}		(ba3);
\path		(b-ba)		edge	[bend left]		node[above]	{}		(ba3);
\path		(b-bb)		edge	[bend right]	node[above]	{}		(b3);

\tikzset{
every path/.style ={black}
}

\path		(i-a)			edge	[bend left]		node[above]	{}		(a3);
\path		(i-b)			edge	[bend left]		node[above]	{}		(b3);
\path		(a-aa)		edge	[bend left]		node[above]	{}		(aa3);
\path		(a-ab)		edge	[bend left]		node[above]	{}		(ba3);
\path		(b-ba)		edge	[bend left]		node[above]	{}		(ba3);
\path		(b-bb)		edge	[bend right]	node[above]	{}		(b3);

\end{pgfonlayer}

\end{tikzpicture}

\caption{The $1$-synthesis game corresponding to the union of $\mathcal{D}_1$ and $\mathcal{D}_2$ (\emph{cf} Figure~\ref{fig:ex}).}
\label{fig:1delay}

\end{figure}

\begin{proposition}\label{prop:k-delay}
    Let $S$ be a specification defined by some transducer $\tra$.
    If Eve wins the $K$-synthesis game $G_{\tra,K}$ for some $K$, then $S$ is
    realisable by a sequential transducer.
\end{proposition}

\begin{proof}[Sketch] If Eve wins the $K$-synthesis game, then, since it is a safety
  game, she can win with a \emph{positional} strategy. Thus, her actions only depend on
  the last visited vertex.
This allows to reconstruct a realiser for $S$, whose states are the possible
vertices of Adam visited by the strategy.
Then, when Adam chooses an input
symbol $\sigma$ in a vertex $(C,D)$ and Eve decides to go to some
vertex $(E,F)$ from $(C,D,\sigma)$, then in the realiser, we add a transition from $(C,D)$ to
$(E,F)$ on $\sigma$, outputting the safe output of $D$ for $\sigma$.
\end{proof}

\subparagraph{Synthesis algorithm} It is worth
noting that the synthesis game allows for a synthesis procedure: for
ascending values of $K$, test whether Eve wins the $K$-synthesis
game (this can be done in \textsf{PTime} in the size of the game). 
If it is the case, then by Proposition~\ref{prop:k-delay} the
specification is realisable, and we can even extract an
implementation corresponding to a winning strategy of Eve. If
it is not the case, then increment $K$ and try again, until $K$
reaches some given upper bound $B$. The $K$-synthesis game is
exponentially large in general (in the transducer defining the
specification, and in $K$). Solving this game efficiently,
using for instance symbolic methods, as done for LTL
synthesis in the synchronous case~\cite{conf/cav/Ehlers10,FiliotSTTT11}, is beyond the scope of this
paper, but is an interesting research direction.

This algorithm is not complete in general: it is shown for instance in
\cite{DBLP:conf/icalp/FiliotJLW16} that some specifications defined by transducers are realisable
by sequential transducers while Eve has no winning strategy in $G_{\tra,K}$
for any $K$. Still, the converse of Proposition~\ref{prop:k-delay} holds for some
subclasses of specifications $\inter{\tra}$. For example, in the synchronous
setting, where we want to synthesise a synchronous
sequential transducer, it suffices to take $K=0$. This gives an
\textsf{ExpTime} procedure to check the realisability of $\inter{\tra}$ by 
a synchronous sequential transducer. If $\tra$ is
\emph{finite-valued}, then by taking $K$ large
enough (triply exponential in $\tra$), we get completeness
\cite{DBLP:conf/icalp/FiliotJLW16}. Finally, if $\tra$ is functional, then Eve wins $G_{\tra,K}$
for some $K$ iff $\tra$ is equivalent to a sequential transducer,
and a polynomial $K$ (in $\tra$) suffices~\cite{BealCPS03}. 

In this paper, we obtain completeness for multi-sequential
specifications by taking $K$ exponential in $\tra$
(Proposition~\ref{prop:boundeddelay}). While this allows us to decide
realisability using the game approach, the time complexity will not be optimal
(\textsf{2ExpTime}). We indeed devise, in Section~\ref{sec:async}, a 
\textsf{PSpace} realisability-checking procedure based on an effective
characterisation of realisable multi-sequential specifications. If the \textsf{PSpace} procedure concludes that the
specification is realisable, one can run the former game-solving procedure to
synthesise a realiser, for ascending values of $K$. This way, one may hope to
synthesise a ``small'' realiser.

\section{The asynchronous setting}\label{sec:async} 

We first characterise recursively the multi-sequential specifications
which are sequentially realisable (Theorem~\ref{thm:charac}). 
Then, we provide an equivalent
characterisation, non-recursive and easier to check algorithmically, but more
technical. 

Similarly to the synchronous case, we define a notion of critical situation to
which a realiser must react. In the synchronous case, it was just a prefix on
which at least two sequential transducers were producing different outputs. In
the asynchronous case, two sequential transducers may produce different outputs
on the same prefix, but this may not be problematic in the case where one is
ahead of the other, i.e., the output of one run is a prefix of the output of the other.
A critical situation is rather a prefix where the
delays between all the outputs of the sequential transducers are too large.
Since no bound is known a priori to define ``too large'', we formalise a
critical situation as a prefix of the form $uv$, such that at least two
sequential transducers loop on $v$, and have a different delay before and after
the loop. By iterating this loop, i.e. by taking a prefix $uv^n$, the delay
between these two transducers will grow unboundedly when $n$ increases. For such
loops, the situation will get critical if a realiser does not react.

\begin{definition}[critical loop]\label{def:critloop}
Let $\tra = \biguplus_{i=1}^n \mathcal{D}_i$ be an $n$-sequential
transducer. A \emph{critical loop} for $\tra$ is
a triple $(u,v,\mathcal{X})\in \Sigma_\inp^*\times
\Sigma_\inp^*\times 2^{\{1,\dots,n\}}$ such that
\begin{enumerate}
\item\label{init}
for all $i\in\mathcal{X}$, there exists an initial run $p_i\xrightarrow{u|\alpha_i} q_i\xrightarrow{v|\beta_i}q_i$
of $\mathcal{D}_i$ on $uv$;
\item\label{void}
for all $i\in\{1,\dots,n\}\setminus\mathcal{X}$,
there is no run of $\mathcal{D}_i$ on $u$;
\item\label{diff}
There exists $i,j\in\mathcal{X}$ such that
$\delay(\alpha_i,\alpha_j) \neq \delay(\alpha_i\beta_i,\alpha_j\beta_j)$.
\end{enumerate}
\end{definition}

Our characterisation echoes the one of the synchronous
setting. It says that whenever there is a critical situation (a
critical loop), a realiser must be able to drop some of the sequential
transducers, in order to prevent the delays to grow unboundedly, while
preserving the residual domain. Formally:

\begin{theorem}[recursive characterisation]\label{thm:charac}
    Let $\tra = \biguplus_{i=1}^n\mathcal{D}_i$ be a multi-sequential transducer over
    $\Sigma_\inp,\Sigma_\outp$. Then $\inter{\tra}$ is realisable by a
    sequential transducer, iff, for all critical loops $(u,v,\mathcal{X})$, there
    exists $\mathcal{Y}\subsetneq \mathcal{X}$ such that 
    \begin{enumerate}
      \item $\forall i,j\in \mathcal{Y}$, $\delay(\alpha_i,\alpha_j) =
        \delay(\alpha_i\beta_i,\alpha_j\beta_j)$ (following the notations of Definition~\ref{def:critloop}),
      \item $u^{-1}\dom(\tra) = \bigcup_{i\in \mathcal{Y}}
        u^{-1}\dom(\mathcal{D}_i)$,
      \item $\bigcup_{i\in \mathcal{Y}} (u,\ell)^{-1}\inter{\mathcal{D}_i}$
        is realisable by a sequential transducer, where
        $\ell = \bigwedge_{i\in\mathcal{X}} \alpha_i$.  
    \end{enumerate}
\end{theorem}

\begin{proof}[Sketch]
    $\Rightarrow$ Let $\mathcal{U}$ be a sequential transducer
    realising $\inter{\tra}$. For every critical loop $(u,v,\mathcal{X})$ of $\tra$,
    the corresponding set $\mathcal{Y}$ is obtained
    by getting rid of all the transducers that stray arbitrarily far from
    $\mathcal{U}$ on the input words of the form $uv^*$.
    Then, the first property is immediate, and the other two follow from
    the fact that $\mathcal{U}$ is sequential and realises $\inter{\tra}$. 

    $\Leftarrow$ Conversely, assuming that whenever a critical loop is met there exists
    a set $\mathcal{Y}$ satisfying the three conditions, we prove by induction on
    the degree $n$ of sequentiality of $\inter{\tra}$ that Eve has
    a winning strategy in the $K_\tra$-synthesis game for $\tra$, for some
    well-chosen value $K_\tra$ depending on $\tra$. By Proposition~\ref{prop:k-delay}, this entails the existence of a
    sequential realiser.
    
    Note that in the synthesis game, since $\tra$ is a union of sequential
    transducers, for each accessible vertex $(C,D)$ of Adam, and for every
    $i\in\{1,\dots,n\}$, there is at most one state $q_i$ of
    $\mathcal{D}_i$ occurring in $live(D)$, and if there exists such a state,
    $|D(q_i)| = 1$. As a consequence, Eve's
    actions in the synthesis game, which consist in dropping a subset of
    pairs (state,delay),  actually correspond here to ``dropping'' a
    subset of sequential transducers.
    
    If $n=1$, then $\tra$ is sequential, and the strategy of Eve that consists in never dropping $\tra$
    is winning. Now, suppose that $n>1$. 
    In order to demonstrate that Eve has a winning strategy, we show that for every input word chosen by
    Adam, either Eve can keep track of all the
    transducers in the $K_{\tra}$-synthesis game, which ensures her a win, or she can drop some
    transducers on the way, while reaching a state from which she has a winning strategy.
    
    Let $u \in \Sigma_{\inp}^*$, and let $(C_0,D_0)$ be the state reached by Eve on input $u$ if she drops nothing.
    If $(C_0,D_0)$ is not part of the $K_{\tra}$-synthesis game, i.e., for some $q \in C_0$, $D_0(q) = \{ w \}$ with $|w| > K_{\tra}$,
    this implies the existence of a decomposition $u_1u_2u_3$ of $u$ such that
    $(u_1,u_2,\mathcal{X})$ is a critical loop for some $\mathcal{X} \subseteq \{1, \ldots, n\}$.
    Then, by hypothesis, there exists a
    subset $\mathcal{Y}\subsetneq \mathcal{X}$ which satisfies the three
    conditions of the theorem, hence $\tra' = \biguplus_{i\in\mathcal{Y}}(u_1,\ell)^{-1}\inter{\mathcal{D}_i}$ is realisable
    by a sequential transducer.
    In particular, $\tra'$ satisfies the 
    conditions on critical loops (implication $\Rightarrow$ shown before), and, by the induction hypothesis (since $\tra'$
    is $|\mathcal{Y}|$-sequential and $|\mathcal{Y}|<n$), Eve has a
    winning strategy in the $K_{\tra'}$-synthesis game for $\tra'$ from
    the initial vertex. Lifting this strategy to the $K_\tra$-synthesis game for $\tra$
    yields a winning strategy for Eve from the state $(C,D')$, where $(C,D)$ is the state reached by Eve
    on input $u_1u_2$ if she drops nothing, and $D'$ is obtained from $D$ by dropping
    all the transducers that are not part of $\mathcal{Y}$.
\end{proof}

The proof of Theorem~\ref{thm:charac} shows that if
a multi-sequential specification is realisable, Eve wins the $K$-synthesis game for $K$ computable from the specification, as stated in Proposition~\ref{prop:boundeddelay}.
As explained in Section~\ref{sec:delay}, solving the $k$-synthesis game for
ascending values of $k$ then provides a practical way to synthesise a realiser,
but the complexity is not optimal. 

\begin{proposition}[bounded delay]\label{prop:boundeddelay}
    Let $S$ be a specification defined by an $n$-sequential
    transducer $\tra$. Then $S$ is
    realisable by some sequential transducer iff Eve wins the
    $K$-synthesis game for $K = L(6M)^{n^2}$, where $L$ is the longest
    output occurring on a transition of $\tra$, and $M$ is the maximal
   number of states of a sequential transducer of $\tra$.  
\end{proposition}

\begin{theorem}\label{thm:complexity_synthesis_asynchronous}
    A realisable specification $S$ defined by a multi-sequential transducer $\tra$ with $m$ states
    admits a realiser of size doubly exponential in $m$.
    Moreover, there exists a family $(S_n)_{n \in \mathbb{N}}$ of realisable specifications such that
    for every $n \in \mathbb{N}$, $S_n$ is definable by a multi-sequential transducer of size
    polynomial in $n$, and every sequential transducer realising $S_n$ has a size that is
    doubly exponential in $n$.
\end{theorem}

\begin{proof}
Let $S \subset \Sigma^* \times \Gamma^*$ be a realisable specification defined by an $n$-sequential transducer $\tra$ with
a set of states $Q$ of size $m$.
Note that $n \leq m$, hence, by Proposition \ref{prop:boundeddelay},
Eve wins the $K$-synthesis game for some $K$ exponential in $m$.
Then, the construction presented in the proof of Proposition \ref{prop:k-delay}
exposes a realiser whose set of states $Q'$ consists of Adam's vertices that are reachable in the $K$-synthesis game.
For every such vertex $(C,D) \in 2^Q \times (Q \rightarrow 2^{\Gamma^*})$, since $\mathcal{T}$ is $n$-sequential, 
there is at most $n$ sates $q \in Q$ satisfying $D(q) \neq \emptyset$.
Moreover, for every such state we have $D(q) = \{ w \}$ for some $w \in \Gamma^*$ satisfying $|w| \leq K$.
Therefore, the size of $Q'$ is bounded by
$2^{m}(m(|\Gamma|^{K+1}))^{n}$, which is doubly exponential in $m$.

In order to expose the family $(S_n)_{n \in \mathbb{N}}$,
we use the notion of $j$-pairs, presented in \cite{DBLP:journals/corr/KleinZ14}.
For every $n \in \mathbb{N}$, let us consider the alphabet $\alphint{n} = \{1, \ldots, n\}$.
A \emph{bad $j$-pair} of a word $u = i_1 \ldots i_m \in \alphint{n}^*$ is a pair of positions $1 \leq k < k' \leq m$ such that
$i_k = i_{k'} = j$, and for all $k < \ell < k'$, $i_\ell \leq j$.
Then every $u \in \alphint{n}^*$ satisfying $|u| \geq 2^n$ admits a bad $j$-pair for some $1 \leq j \leq n$,
and there exists a word, denoted by $\nojp{n}$, that has size $2^n-1$, and contains no $j$-pair (see \cite{DBLP:journals/corr/KleinZ14}).
We now consider the finite alphabet $\Sigma = \{a,b\}$.
For every $n \in \mathbb{N}$, let $\prodalp{n}$ denote the alphabet $\alphint{n} \times \Sigma$.
We denote by $\pi_1 : \prodalp{n}^* \rightarrow \alphint{n}^*$ and  $\pi_2 : \prodalp{n}^* \rightarrow \Sigma^*$
the projections on the first, respectively second component. 
Let $\swap : \prodalp{n}^* \rightarrow \Sigma^*$ be the function mapping $w \in \prodalp{n}^*$ to the word obtained
by taking the last letter of $\pi_2(w)$ and putting it at the beginning, i.e., $\swap(w) = \sigma v$ where $\sigma \in \Sigma$ and $v \in \Sigma^*$ satisfy $\pi_2(w) = v \sigma$.
We consider the specification 
\[
S_n = \{(w,\swap(w))| w \in \prodalp{n}^* \}  \cup \{ (w, \epsilon) | w \in \prodalp{n}^* \textup{ contains a bad $j$-pair for some $1 \leq j \leq n$} \}.
\]
Then $S_n$ is definable by an $(n+2)$-sequential transducers with $3(n+2)$ states, since
the function $\swap$ is definable by the union of $2$ sequential transducers of size $3$,
and for every $1 \leq j \leq n$, the set of words $w \in \alphint{n}^*$ containing a bad $j$-pair is recognisable by a deterministic automaton of size $3$.
Moreover, since every word $u \in \alphint{n}^*$ of size greater than $2^n$ admits a bad $j$-pair for some $j$,
$S_n$ is realised by the sequential transducer mapping every word $w \in \prodalp{n}^*$ satisfying $|w| < 2^{n}$
to $\swap(w)$, and every $w \in \prodalp{n}^*$ satisfying $|w| \geq 2^{n}$ to $\epsilon$.

We now show that every sequential transducer $\mathcal{D}$ realising $S_n$ has at least $2^{2^{n}-1}$ states.
Let $\mathcal{D} = ((\Sigma,Q,I,F,\Delta),\rho,\tau)$ be a sequential transducer realising $S_n$.
For every $v \in \Sigma^*$ such that $|v| = 2^n-1$, let $\produ{v} \in \prodalp{n}^*$ denote the word satisfying $\pi_1(\produ{v}) = \nojp{n}$ and $\pi_2(\produ{v}) = v$.
We now show that for every pair of distinct words $v_1,v_2 \in \Sigma^*$ of size $2^{n}-1$,
the states reached by $\mathcal{D}$ on input $\produ{v_1}$ and $\produ{v_2}$ are distinct.
This allows us to conclude the proof, since $\Sigma^*$ contains $2^{2^{n}-1}$ such words.
Given $v \in \Sigma^*$ satisfying $|v| = 2^{n}-1$, let $\rho_v: p_0 \yrightarrow{ \produ{v} | v'}[-1pt] p_{v}$ denote the accepting run of $\mathcal{D}$ on input $\produ{v}$.
Then $v' = \epsilon$, since if the first letter of $v'$ was an $a$, $\mathcal{D}$ would not be able to produce an acceptable output
on input $\produ{v} \cdot (1,b)$, and a similar contradiction would be reached if the first letter of $v'$ was a $b$.
Therefore, the output associated to $\produ{v}$ is produced by the terminal function of $\mathcal{D}$,
i.e., $\tau(p_{v}) = \swap(\produ{v})$.
Since $\swap$ is injective, for every pair of distinct words $v_1,v_2 \in \Sigma^*$ of size $2^n-1$, $p_{v_1} \neq p_{v_2}$.
\end{proof}



We are now ready to show how to decide the realisability of multi-sequential
specifications in \textsf{PSpace}. Consider the characterisation
given in Theorem~\ref{thm:charac}. We rely on the notion of
witness for the non-satisfaction of this characterisation, and we show how to decide the existence of a witness, using a reduction to
the emptiness of reversal-bounded counter machines. 

The notion of witness intuitively consists in the following ingredients: (1) an 
unfolding (modeled as a tree) of the recursive characterisation of
Theorem~\ref{thm:charac} and (2) an explicit formulation of delay differences
using simple properties of words. Formally, given an $n$-sequential transducer
$\tra = \biguplus_{i=1}^n \mathcal{D}_i$, where each $\mathcal{D}_i$ is sequential, a
\emph{witness} for $\tra$ is a finite tree $t$ whose nodes are labelled in
$\Sigma_\inp^*\times \Sigma_\inp^*\times (2^{\{1,\dots,n\}}\setminus
\{\varnothing\})$. For any node $x$ of $t$, we denote by $(u_x,v_x,S_x)$
its label. For all nodes
$x,y,z$ of $t$, it is required that: 
\begin{enumerate}
\item (maximality) if $x$ is the root, $S_x = \{1,\dots,n\}$;
\item (consistency) $S_x$ can be split into two disjoint sets
  $N_x,L_x$ such that for all $i\in N_x$  there
  is no run of $\mathcal{D}_i$ on $u_x$, and for all $i\in L_x$ 
  there is a run of 
  $\mathcal{D}_i$ on $u_xv_x$ from its initial state $q_0^i$, of the
  form $q_0^i\yrightarrow{ u_{x}| \alpha_{x,i}}[-1pt] p_{x,i}
  \yrightarrow{v_{x}|\beta_{x,i}}[-1pt] p_{x,i}$;
\item (monotonicity) if $y$ is a child of $x$, then
  $S_y\subsetneq L_x$ and $u_x$ is a prefix of $u_y$;
\item (partition) if $Y$ is the set of children of $x$, then $\{ S_y\mid y\in Y\}$ partitions $L_x$;
\item (delays) if $y$ and $z$ are different children of $x$,
  for all $i\in S_y$ and $j\in
  S_z$,  either $|\beta_{x,i}|\neq |\beta_{x,j}|$ or, $\beta_{x,i}\beta_{x,j}\neq\epsilon$ 
   and, $\alpha_{x,i}$ and $\alpha_{x,j}$ mismatch\footnote{Two words $u,v$
    mismatch if there is a position $i$ such that $i\leq |u|,|v|$ and the $i$th letter
    of $u$ differs from the $i$th letter of $v$.}; 
\item (leaves) if $x$ is a leaf,  then  there is $w\in\Sigma_\inp^*$ such
  that 
  $u_{x}w \in \dom(\tra)$ and $u_{x}w\not\in \dom(\mathcal{D}_i)$ for
  all $i\in S_x$.
\end{enumerate}
 Intuitively, conditions $2$ and $5$ require that the words $u_x,v_x$ are 
 critical loops. The delay difference required in the definition of
 critical loops is not explicit here, but
 rather replaced by simple properties of words (condition 5), which
 are easier to check algorithmically. These properties are  not
 strictly equivalent to delay difference, but up to iterating the loop
 on $v_x$ a sufficient number of times, they are. 
 Conditions $1,3,4$ correspond to properties of the subsets met
when unfolding the recursive characterisation of
Theorem~\ref{thm:charac}. They also allow us to bound linearly the number
of nodes of a witness. As announced, all these conditions characterise
 the unrealisable multi-sequential specifications:

\begin{lemma}\label{coro:charac2}
    A multi-sequential specification defined by a trim transducer $\tra$ is
    not realisable by a sequential transducer if and only if there exists a
    witness for $\tra$.
\end{lemma}


\begin{theorem}\label{thm:complexityasynchronous}
    The realisability problem by some sequential transducer of a specification
    defined by a multi-sequential transducer is \textsf{PSpace-c}. 
\end{theorem}

\begin{proof}[Sketch]
    \textsf{PSpace}-hardness has been shown in Section~\ref{sec:PSpaceHardness}.
    To show \textsf{PSpace}-easyness, we reduce the problem to deciding the
    emptiness of the language of a counter machine, whose counters make at most $1$
    reversal (i.e. move from increasing to decreasing mode).
    This is known to be in \textsf{NLogSpace}
    \cite{DBLP:journals/jcss/GurariI81}. Our machine is exponentially large (in
    the transducer defining the specification), but can be constructed on the
    fly, hence we get \textsf{PSpace}.

    A bit more precisely, we first define the notion of \emph{skeleton} $s$,
which is a witness without the words $u_x,v_x$, hence there are finitely
many skeletons, each one of polynomial size. Given an enumeration $x_1\dots x_n$
in depth-first order of the nodes of $s$, we construct a counter machine $M_s$
which recognises sequences of the form
    $x_1 w_{x_1}\# v_{x_1} \dots  x_n w_{x_n} \# v_{x_n}$ such that if
    we extend any label of a node $x$ in $s$ with the pair of words
    $(w_{y_1}\dots w_{y_k}, v_x)$, where $y_1\dots y_k$ is the path
    from the root to $x$, we get a witness. Hence, there exists a
    witness iff there exists a skeleton $s$ such that $L(M_s)$ is
    non-empty. Our algorithm non-deterministically guesses a skeleton
    and runs a procedure to check in \textsf{PSpace} the emptiness of
    $M_s$. 

    Let us intuitively explain how $M_s$ works. Conditions $1,3,4$ and
    $6$ are regular, so no counter is needed there. %
    Counters are only necessary to check Condition $5$, for instance to compute
the length of the words $\beta_{x,i}$, and to check the existence of a mismatch
between a word $\alpha_{x,i}$ and a word $\alpha_{x,j}$. First, a mismatch
position $m$ is guessed, by incrementing for some time two counters $c_{i,x}$
and $c_{j,x}$ in parallel. Then, they are decremented according to
the length of outputs produced by simulating the transitions of $\mathcal{D}_i$
and $\mathcal{D}_j$ respectively. When one of them reaches $0$, say $c_{i,x}$,
we store the $m$th symbol of the output of $\mathcal{D}_i$ on $u_x$ in memory.
We do the same for $c_{j,x}$ and later on check that the two stored symbols are
different. 
\end{proof}

\section{Conclusion} We have identified a class of 
specifications (whose membership is decidable in {\sf PTime}), for which the sequential
realisability problem is {\sf PSpace-c}, both in the asynchronous and
synchronous settings. This is in contrast to the general case, which
is {\sf ExpTime-c} for synchronous specifications, and undecidable in
the asynchronous case. Our procedure
allows to synthesise a sequential transducer whenever the
specification is realisable, and allows for incremental testing, via
the solvability of a two-player game parameterised by the longest
output allowed to be queued by a realiser before being output.

While the class of multi-sequential specifications is natural, as the closure of
graphs of sequential functions under finite unions, we believe that it may also
be interesting for practical applications. In particular, Vardi and Lustig have
defined the concept of synthesis from component
libraries~\cite{journals/sttt/LustigV13}, in the synchronous setting, over
infinite words. In this setting, given a set of components (synchronous
sequential transducers over finite words), a specification $S$ over infinite
words, the question is whether the components can be arranged in such a way
which realises the specification (by linking the final states of the components
to the initial state of another component). This problem was shown to be decidable.
We would like to investigate another way of reusing existing components, which
is tightly related to multi-sequential specifications: given components
$C_1,\dots,C_n$ represented as sequential transducers and a specification $S$,
decide whether there exists a sequential function $f$ such that $f$ and $S$ have
the same domain, $f\subseteq \bigcup_i C_i$ and $f$ satisfies $S$. This is
beyond the scope of this paper but we plan to investigate further this question
in the near future.

\bibliography{biblio}

\begin{thebibliography}{10}

\bibitem{DBLP:journals/jacm/AlurHK02}
Rajeev Alur, Thomas~A. Henzinger, and Orna Kupferman.
\newblock Alternating-time temporal logic.
\newblock {\em Journal of the {ACM}}, 49(5):672--713, 2002.
\newblock URL: \url{http://doi.acm.org/10.1145/585265.585270}.

\bibitem{BealCPS03}
Marie-Pierre B{\'e}al, Olivier Carton, Christophe Prieur, and Jacques
  Sakarovitch.
\newblock Squaring transducers: an efficient procedure for deciding
  functionality and sequentiality.
\newblock {\em Theoretical Computer Science}, 292(1):45--63, 2003.

\bibitem{berstel2009}
Jean Berstel and Luc Boasson.
\newblock Transductions and context-free languages.
\newblock {\em Ed. Teubner}, pages 1--278, 1979.

\bibitem{CarayolL14}
Arnaud Carayol and Christof L{\"o}ding.
\newblock {Uniformization in Automata Theory}.
\newblock In {\em Proceedings of the 14th Congress of Logic, Methodology and
  Philosophy of Science Nancy, July 19-26, 2011}, pages 153--178, London, 2014.
  College Publications.

\bibitem{DBLP:journals/iandc/ChatterjeeHP10}
Krishnendu Chatterjee, Thomas~A. Henzinger, and Nir Piterman.
\newblock Strategy logic.
\newblock {\em Information and Computation}, 208(6):677--693, 2010.
\newblock URL: \url{https://doi.org/10.1016/j.ic.2009.07.004}.

\bibitem{DBLP:conf/stacs/ChoffrutS86}
Christian Choffrut and Marcel~Paul Sch{\"{u}}tzenberger.
\newblock D{\'{e}}composition de fonctions rationnelles.
\newblock In {\em 2nd Annual Symposium on Theoretical Aspects of Computer
  Science, STACS}, pages 213--226, 1986.

\bibitem{Chur62}
{Church, Alonzo}.
\newblock Logic, arithmetic and automata.
\newblock In {\em International Congress of Mathematics}, pages 23--35,
  Stockholm, 1962.

\bibitem{DBLP:conf/fossacs/DaviaudJRV17}
Laure Daviaud, Isma{\"e}l Jecker, Pierre-Alain Reynier, and Didier Villevalois.
\newblock Degree of sequentiality of weighted automata.
\newblock In Javier Esparza and Andrzej~S. Murawski, editors, {\em Proceedings
  of the 20th International Conference on Foundations of Software Science and
  Computation Structures, FOSSACS 2017, Uppsala, Sweden, April 22-29}, pages
  215--230. Springer Berlin Heidelberg, 2017.
\newblock URL: \url{https://doi.org/10.1007/978-3-662-54458-7_13}.

\bibitem{DBLP:conf/concur/AlfaroHM01}
Luca de~Alfaro, Thomas~A. Henzinger, and Rupak Majumdar.
\newblock Symbolic algorithms for infinite-state games.
\newblock In {\em Proceedings of the 12th International Conference in
  Concurrency Theory, CONCUR 2001, Aalborg, Denmark, August 20-25}, pages
  536--550, 2001.
\newblock URL: \url{https://doi.org/10.1007/3-540-44685-0_36}.

\bibitem{conf/cav/Ehlers10}
R{\"u}diger Ehlers.
\newblock Symbolic bounded synthesis.
\newblock In {\em Proceedings of the 22nd International Conference on Computer
  Aided Verification, CAV 2010, Edinburgh, {UK}, July 15-19}, volume 6174 of
  {\em Lecture Notes in Computer Science}, pages 365--379. Springer, 2010.

\bibitem{Eilenberg:1974:ALM:540337}
Samuel Eilenberg.
\newblock {\em Automata, Languages, and Machines}.
\newblock Academic Press, 1974.

\bibitem{DBLP:conf/icalp/FiliotJLW16}
Emmanuel Filiot, Isma{\"{e}}l Jecker, Christof L{\"{o}}ding, and Sarah Winter.
\newblock On equivalence and uniformisation problems for finite transducers.
\newblock In {\em 43rd International Colloquium on Automata, Languages, and
  Programming, ICALP 2016, July 11-15, Rome, Italy}, pages 125:1--125:14, 2016.
\newblock URL: \url{https://doi.org/10.4230/LIPIcs.ICALP.2016.125}.

\bibitem{FiliotSTTT11}
Emmanuel Filiot, Naiyong Jin, and Jean-Fran\c{c}ois Raskin.
\newblock Exploiting structure in {LTL} synthesis.
\newblock {\em International Journal on Software Tools for Technology
  Transfer}, 2011.

\bibitem{FiliotJR11}
Emmanuel Filiot, Naiyong Jin, and Jean{-}Fran{\c{c}}ois Raskin.
\newblock Antichains and compositional algorithms for {LTL} synthesis.
\newblock {\em Formal Methods in System Design}, 39(3):261--296, 2011.

\bibitem{DBLP:conf/csl/FridmanLZ11}
Wladimir Fridman, Christof L{\"{o}}ding, and Martin Zimmermann.
\newblock Degrees of lookahead in context-free infinite games.
\newblock In {\em Computer Science Logic, 25th International Workshop / 20th
  Annual Conference of the EACSL, {CSL} 2011, September 12-15, 2011, Bergen,
  Norway, Proceedings}, pages 264--276, 2011.
\newblock URL: \url{https://doi.org/10.4230/LIPIcs.CSL.2011.264}.

\bibitem{DBLP:journals/jcss/GurariI81}
Eitan~M. Gurari and Oscar~H. Ibarra.
\newblock The complexity of decision problems for finite-turn multicounter
  machines.
\newblock {\em Journal of Computer and System Science}, 22(2):220--229, 1981.
\newblock URL: \url{https://doi.org/10.1016/0022-0000(81)90028-3}.

\bibitem{conf/fossacs/HoltmannKT10}
Michael Holtmann, Lukasz Kaiser, and Wolfgang Thomas.
\newblock Degrees of lookahead in regular infinite games.
\newblock In C.-H.~Luke Ong, editor, {\em Proceedings of the 13th International
  Conference on Foundations of Software Science and Computational Structures,
  FOSSACS 2010, Paphos, Cyprus, March 20-28}, volume 6014 of {\em Lecture Notes
  in Computer Science}, pages 252--266. Springer, 2010.

\bibitem{DBLP:journals/sttt/JacobsBBEHKPRRS17}
Swen Jacobs, Roderick Bloem, Romain Brenguier, R{\"{u}}diger Ehlers, Timotheus
  Hell, Robert K{\"{o}}nighofer, Guillermo~A. P{\'{e}}rez,
  Jean{-}Fran{\c{c}}ois Raskin, Leonid Ryzhyk, Ocan Sankur, Martina Seidl,
  Leander Tentrup, and Adam Walker.
\newblock The first reactive synthesis competition {(SYNTCOMP} 2014).
\newblock {\em {STTT}}, 19(3):367--390, 2017.
\newblock URL: \url{https://doi.org/10.1007/s10009-016-0416-3}.

\bibitem{DBLP:conf/dlt/JeckerF15}
Isma{\"{e}}l Jecker and Emmanuel Filiot.
\newblock Multi-sequential word relations.
\newblock In {\em Proceedings of the 19th International Conference on
  Developments in Language Theory, {DLT} 2015, Liverpool, UK, July 27-30},
  pages 288--299, 2015.
\newblock URL: \url{https://doi.org/10.1007/978-3-319-21500-6_23}.

\bibitem{Jobstm07c}
B.~Jobstmann, S.~Galler, M.~Weiglhofer, and R.~Bloem.
\newblock Anzu: {A} tool for property synthesis.
\newblock In {\em Computer Aided Verification, CAV}, pages 258--262, 2007.

\bibitem{BuLa69}
{J.R. B\"{u}chi} and {L.H. Landweber}.
\newblock Solving sequential conditions finite-state strategies.
\newblock {\em Transactions of the American Mathematical Society},
  138:295--311, 1969.

\bibitem{DBLP:journals/corr/KleinZ14}
Felix Klein and Martin Zimmermann.
\newblock How much lookahead is needed to win infinite games?
\newblock {\em Logical Methods in Computer Science}, 12(3), 2016.
\newblock URL: \url{https://doi.org/10.2168/LMCS-12(3:4)2016}.

\bibitem{DBLP:conf/fsttcs/KleinZ16}
Felix Klein and Martin Zimmermann.
\newblock Prompt delay.
\newblock In {\em 36th {IARCS} Annual Conference on Foundations of Software
  Technology and Theoretical Computer Science, {FSTTCS} 2016, December 13-15,
  Chennai, India}, pages 43:1--43:14, 2016.
\newblock URL: \url{https://doi.org/10.4230/LIPIcs.FSTTCS.2016.43}.

\bibitem{journals/iandc/Kobayashi69}
Kojiro Kobayashi.
\newblock Classification of formal languages by functional binary
  transductions.
\newblock {\em Information and Control}, 15(1):95--109, July 1969.

\bibitem{conf/focs/Kozen77}
Dexter Kozen.
\newblock Lower bounds for natural proof systems.
\newblock In {\em FOCS}, pages 254--266. IEEE Computer Society, 1977.
\newblock URL: \url{http://dblp.uni-trier.de/db/conf/focs/focs77.html#Kozen77}.

\bibitem{journals/sttt/LustigV13}
Yoad Lustig and Moshe~Y. Vardi.
\newblock Synthesis from component libraries.
\newblock {\em STTT}, 15(5-6):603--618, 2013.

\bibitem{PnuRos:89}
A.~Pnueli and R.~Rosner.
\newblock On the synthesis of a reactive module.
\newblock In {\em ACM Symposium on Principles of Programming Languages, POPL}.
  ACM, 1989.

\bibitem{springerlink:10.1007/978-3-540-75596-8_33}
Sven Schewe and Bernd Finkbeiner.
\newblock Bounded synthesis.
\newblock In {\em Automated Technology for Verification and Analysis}, volume
  4762 of {\em Lecture Notes in Computer Science}, pages 474--488. Springer
  Berlin Heidelberg, 2007.

\bibitem{Thomas08}
Wolfgang Thomas.
\newblock Church's problem and a tour through automata theory.
\newblock In {\em Pillars of Computer Science, Essays Dedicated to Boris (Boaz)
  Trakhtenbrot on the Occasion of His 85th Birthday}, volume 4800 of {\em
  Lecture Notes in Computer Science}, pages 635--655. Springer, 2008.

\bibitem{DBLP:journals/ita/Zimmermann16}
Martin Zimmermann.
\newblock Delay games with {WMSO+U} winning conditions.
\newblock {\em {RAIRO} - Theoretical Informatics and Applications},
  50(2):145--165, 2016.
\newblock URL: \url{https://doi.org/10.1051/ita/2016018}.

\bibitem{abs-1709-03539}
Martin Zimmermann.
\newblock Finite-state strategies in delay games.
\newblock In {\em Proceedings 8th International Symposium on Games, Automata,
  Logics and Formal Verification, GandALF 2017, Roma, Italy, 20-22 September},
  pages 151--165, 2017.
\newblock URL: \url{https://doi.org/10.4204/EPTCS.256.11}.

\bibitem{Zimmermann17}
Martin Zimmermann.
\newblock Games with costs and delays.
\newblock In {\em 32nd Annual {ACM/IEEE} Symposium on Logic in Computer
  Science, {LICS} 2017, Reykjavik, Iceland, June 20-23}, pages 1--12, 2017.
\newblock URL: \url{https://doi.org/10.1109/LICS.2017.8005125}.

\end{thebibliography}

\end{document}